\documentclass{article}

\newcommand{\typeof}{0} %

\usepackage{ifthen}
\newcommand{\lv}[1]{\ifthenelse{\equal{\typeof}{0}}{#1}{}}
\newcommand{\sv}[1]{\ifthenelse{\equal{\typeof}{0}}{}{#1}}
\newcommand{\lvsv}[2]{\ifthenelse{\equal{\typeof}{0}}{#1}{#2}}

\title{The Geometry of Types}
\lvsv{
  \author{Ugo Dal Lago \and Barbara Petit}
  \date{}
}{
  \authorinfo{Ugo Dal Lago \and Barbara Petit}{
    Universit\`a di Bologna \& INRIA
  }{ \{dallago,petit\}@cs.unibo.it}
  \conferenceinfo{Popl'13,} {January 23--25, 2013, Rome, Italy.}
  \CopyrightYear{2013}
  \copyrightdata{}
}

\usepackage[utf8x]{inputenc}
\usepackage{hyperref}
\usepackage{graphicx}
\usepackage{xcolor}
\usepackage{amsmath,amssymb}
\usepackage{mathabx}
\usepackage{qsymbols}
\usepackage{mathrsfs}
\usepackage{proof}
\usepackage[all,color,frame]{xy}
\lv{\usepackage{fullpage}}

\newtheorem{definition}{Definition}[section]  
\newtheorem{theorem}{Theorem}[section]
\newtheorem{lemma}[theorem]{Lemma}
\newtheorem{proposition}[theorem]{Proposition}
\newtheorem{corollary}[theorem]{Corollary}
\newenvironment{proof}{\begin{trivlist}
       \item[\hskip \labelsep {\bfseries Proof.}]}{\hfill $\Box$ \end{trivlist}}

\newenvironment{varitemize}
{
\begin{list}{\labelitemi}
{\setlength{\itemsep}{0pt}
 \setlength{\topsep}{0pt}
 \setlength{\parsep}{0pt}
 \setlength{\partopsep}{0pt}
 \setlength{\leftmargin}{15pt}
 \setlength{\rightmargin}{0pt}
 \setlength{\itemindent}{0pt}
 \setlength{\labelsep}{5pt}
 \setlength{\labelwidth}{10pt}
}}
{
 \end{list} 
}

\newcounter{numberone}

\newenvironment{varenumerate}
{
\begin{list}{\arabic{numberone}.}
{ \usecounter{numberone}
  \setlength{\itemsep}{0pt}
  \setlength{\topsep}{0pt}
  \setlength{\parsep}{0pt}
  \setlength{\partopsep}{0pt}
  \setlength{\leftmargin}{15pt}
  \setlength{\rightmargin}{0pt}
  \setlength{\itemindent}{0pt}
  \setlength{\labelsep}{5pt}
  \setlength{\labelwidth}{15pt}
}}
{
\end{list} 
}
\newtheorem{fact}{Fact}

\definecolor{violet}{RGB}{120,6,250}
\definecolor{orange}{RGB}{250,90,0}
\newcommand{\cone}[1]{\textcolor{violet}{#1}}  
\newcommand{\ctwo}[1]{\textcolor{orange}{#1}}  

\newcommand{\TODO}[1]{\textcolor{orange}{#1}}

\newcommand{\resp}{\textit{resp.} }

\newcommand{\ie}{\text{i.e.}}
\newcommand{\eg}{\text{e.g.}}

\newcommand{\etc}{\text{etc.} }

\newcommand{\midd}{\; \; \mbox{\Large{$\mid$}}\;\;}

\newcommand{\dlpcfn}{\ensuremath{\mathsf{d}\ell\mathsf{PCF_N}}}
\newcommand{\dlpcfv}{\ensuremath{\mathsf{d}\ell\mathsf{PCF_V}}}
\newcommand{\dlpcf}{\ensuremath{\mathsf{d}\ell\mathsf{PCF}}}
\newcommand{\BLL}{\textsf{BLL}}                      
\newcommand{\LL}{\textsf{LL}}                        
\newcommand{\PCF}{\textsf{PCF}}
\newcommand{\CEK}{\textsf{CEK}}
\newcommand{\cek}{\ensuremath{\mathsf{CEK}_{\mathsf{PCF}}}} 
\newcommand{\KAM}{\textsf{KAM}}
\newcommand{\kam}{\ensuremath{\mathsf{KAM}_{\mathsf{PCF}}}} 

\newcommand{\cbn}{\textsc{cbn}}                      
\newcommand{\cbv}{\textsc{cbv}}                      

\renewcommand{\t}{\ensuremath{t}}          
\newcommand{\s}{\ensuremath{s}}            
\renewcommand{\u}{\ensuremath{u}}          
\renewcommand{\v}{\ensuremath{v}}          
            % ...
\newcommand{\ifz}[3]{\ensuremath{
    \mathtt{\ ifz\ }#1\mathtt{\ then\ }#2\mathtt{\ else\ }#3}} 
\newcommand{\fix}[2][x]{\ensuremath{\mathtt{\ fix\ }#1.#2}}    
\newcommand{\nb}[1][n]{\ensuremath{\underline{\mathtt{#1}}}}    
\newcommand{\suc}{\ensuremath{\mathtt{s}}}                     
\newcommand{\pred}{\ensuremath{\mathtt{p}}}                    

\newcommand{\id}{\ensuremath{\mathrm{id}}}          
\newcommand{\add}{\ensuremath{\mathrm{add}}}        
\newcommand{\fpred}{\ensuremath{\mathrm{pred}}}     
\newcommand{\fsucc}{\ensuremath{\mathrm{succ}}}     

\newcommand{\one}{\ensuremath{\mathrm 1}}          
\newcommand{\zero}{\ensuremath{\mathrm 0}}         
\newcommand{\I}{\ensuremath{\mathrm{I}}}           
\newcommand{\J}{\ensuremath{\mathrm{J}}}           
\newcommand{\K}{\ensuremath{\mathrm{K}}}           
\renewcommand{\H}{\ensuremath{\mathrm{H}}}         
\renewcommand{\L}{\ensuremath{\mathrm{L}}}         
\newcommand{\M}{\ensuremath{\mathrm{M}}}           
\newcommand{\N}{\ensuremath{\mathrm{N}}}           
\newcommand{\inb}[1][n]{\ensuremath{\underline{#1}}}  
\newcommand{\f}{\ensuremath{\mathrm{f}}}           
\newcommand{\g}{\ensuremath{\mathrm{g}}}           
\newcommand{\h}{\ensuremath{\mathrm{h}}}           
\renewcommand{\i}{\ensuremath{\mathrm{i}}}           
\renewcommand{\j}{\ensuremath{\mathrm{j}}}           
\renewcommand{\k}{\ensuremath{\mathrm{k}}}           
\renewcommand{\l}{\ensuremath{\mathrm{l}}}           
\newcommand{\m}{\ensuremath{\mathrm{m}}}           
\newcommand{\n}{\ensuremath{\mathrm{n}}}           
\newcommand{\p}{\ensuremath{\mathrm{p}}}           
\newcommand{\q}{\ensuremath{\mathrm{q}}}           
\newcommand{\mnu}{\ensuremath{-}}                  
\newcommand{\fc}[4][b]{\bigotriangleup_{#1}^{#2,#3}#4}        

\newcommand{\isubst}[2][a]{\ensuremath{\{#2/#1\}}} 
\renewcommand{\int}{\ensuremath{\mathbb{N}}}         
\newcommand{\itp}[2][`r]{\ensuremath{\llbracket #2\rrbracket^{\ep}_{#1}}}  

\newcommand{\symb}{\ensuremath{\mathcal{S}}}       
\newcommand{\nsymb}{\ensuremath{\mathcal{N}}}      

\newcommand{\NatPCF}{\ensuremath{\mathtt{Nat}}}    
\newcommand{\toPCF}{\Rightarrow}                   
\newcommand{\tpcfone}{\ensuremath{T}}

\newcommand{\A}{\ensuremath{A}}           
\newcommand{\B}{\ensuremath{B}}           
\newcommand{\C}{\ensuremath{C}}           
\newcommand{\D}{\ensuremath{D}}           
           % ...
\newcommand{\Nat}[2]{\ensuremath{\mathtt{Nat}[#1,#2]}}
\newcommand{\NatU}[1]{\ensuremath{\mathtt{Nat}[#1]}}  
\newcommand{\mtyp}[3][a]{\ensuremath{[#1<#2]\cdot#3}} 
\newcommand{\toLL}{\multimap}             
\newcommand{\marr}[1]{\stackrel{#1}{\multimap}}       
\newcommand{\teq}{\equiv}
\newqsymbol{`<}{\sqsubseteq}              
\newqsymbol{`=}{\equiv}                   

\newcommand{\forget}[1]{\ensuremath{(\![#1]\!)}}      
\newcommand{\typsymb}[2][\polone]{\ensuremath{#2^{#1}}}  

\newcommand{\arrv}[4][a]{\mtyp[#1]{#2}{(#3\toLL#4)}}  

\newcommand{\ep}{\ensuremath{\mathcal{E}}}      
\newcommand{\eptwo}{\ensuremath{\mathcal{F}}}
\newcommand{\rep}{\ensuremath{\mathcal{A}}}     
\newcommand{\reptwo}{\ensuremath{\mathcal{B}}}
\newcommand{\repthree}{\ensuremath{\mathcal{C}}}

\newcommand{\ispec}[2]{\ensuremath{(#1,#2)\text{-specified}}}    
\newcommand{\tspec}[3]{\ensuremath{(#1,#2,#3)\text{-specified}}} 

\newcommand{\fiv}{\ensuremath{`f}}              
\newcommand{\ictx}{\ensuremath{`F}}             
\newcommand{\judg}[7][\ep]{\ensuremath{#2;#3;#4\vdash_{#5}^{#1}#6:#7}} 
\newcommand{\ejudg}[5][\ep]{\ensuremath{#2\vdash_{#3}^{#1}#4:#5}} 
\newcommand{\ijudg}[4][\ep]{\ensuremath{#2;#3\vDash_{#1}#4}}      
\newcommand{\eijudg}[2][\ep]{\ensuremath{\vDash_{#1}#2}}          
\newcommand{\sjudg}[4][\ep]{\ensuremath{#2;#3\vdash_{#1}#4}}     
\newcommand{\cjudg}[3]{\ensuremath{#1;#2\vDash#3\downarrow}}     
\newcommand{\scval}[2][\ep]{\ensuremath{#1\Vdash\bigwedge#2}}    
\newcommand{\ljudg}[4]{\ensuremath{#1\vdash_{#2}#3:#4}}          
\newqsymbol{`0}{\emptyset}                      
\newqsymbol{`|}{\downarrow}                     
\newqsymbol{`=}{:=}                             

\newcommand{\mapping}{\ensuremath{`m}}          
\newcommand{\mappingtwo}{\ensuremath{`e}}          
\newcommand{\sig}[1][]{\ensuremath{`S_{#1}}}    

\newcommand{\env}[1][`x]{\ensuremath{#1}}              
\newcommand{\mkenv}[2][x]{\ensuremath{#1\mapsto#2}}    
\newcommand{\clo}[2][\env]{
  \ensuremath{\langle\,#2\,;\,#1\,\rangle}}            
\newcommand{\cloee}[1]{\ensuremath{\langle\,#1\,\rangle}} 
\newcommand{\cloc}[1][c]{\ensuremath{\mathbf{\sf #1}}} 
\newcommand{\clov}[1][v]{\ensuremath{\mathbf{\sf #1}}} 
\newcommand{\stk}[1][`p]{\ensuremath{#1}}              
\newcommand{\estk}{\ensuremath{\diamond}}              
\newcommand{\stkfun}[2]{\ensuremath{\mathsf{fun}#1`.#2}}  
\newcommand{\stkarg}[2]{\ensuremath{\mathsf{arg}#1`.#2}}  
\newcommand{\stkif}[4][\env]{\ensuremath{
    \mathsf{fork}\,\langle#2\,;\,#3\,;\,#1\rangle`.#4}}               
\newcommand{\stkp}[1]{\ensuremath{\mathsf{p}`.#1}}      
\newcommand{\stks}[1]{\ensuremath{\mathsf{s}`.#1}}      
\newcommand{\proc}[2]{\ensuremath{#1~\star~#2}}         

\newcommand{\tocek}{\ensuremath{\succ}}                 
\newcommand{\vconv}[1][n]{\ensuremath{\Downarrow^{#1}_{_{\sf V}}}} 
\newcommand{\nconv}[1][n]{\ensuremath{\Downarrow^{#1}_{_{\sf N}}}} 

\newcommand{\mainfct}{\ensuremath{\textsc{gen}}}    
\newcommand{\der}{\ensuremath{d_{\sf v}}}               
\newcommand{\derpcf}{\ensuremath{d_{\PCF}}}             
\newcommand{\scond}{\ensuremath{\mathscr{R}}}          
\newcommand{\scondtwo}{\ensuremath{\mathscr{S}}}          
\newcommand{\scondthree}{\ensuremath{\mathscr{U}}}          
\newcommand{\typjudg}{\ensuremath{\mathcal{J}}}        
\newcommand{\progname}{\ensuremath{\mathtt{CheckBound}}}          
\newcommand{\anot}[2][\fiv]{\ensuremath{`a(#1;#2)}}    

\newcommand{\algder}{\textsc{der}}                      
\newcommand{\algderp}[6]{\textsc{der}((#1,#2);#3;#4;#5;#6)}
\newcommand{\algweak}{\textsc{weak}}                     
\newcommand{\algweakp}[4]{\textsc{weak}(#1;#2;#3;#4)}
\newcommand{\algcontr}{\textsc{ctr}}                    
\newcommand{\algcontrp}[9]{\textsc{ctr}((#1,#2,#3);(#4,#5);#6;#7;#8;#9)}
\newcommand{\algdig}{\ensuremath{\textsc{dig}}}        
\newcommand{\algdigp}[9]{\ensuremath{\textsc{dig}((#1,#2);(#3,#4);#5;(#6,#7);#8;#9)}}
\newcommand{\algax}{\textsc{eq}}                       
\newcommand{\algaxp}[5]{\textsc{eq}((#1;#2);#3;#4;#5)}
\newcommand{\algsub}{\textsc{subs}}                      
\newcommand{\algsubp}[7]{\textsc{subs}((#1,#2);#3;#4;#5;#6;#7)}
\newcommand{\algif}{\textsc{fork}}                       
\newcommand{\algifp}[6]{\textsc{fork}(#1;#2;#3;#4;#5;#6)}

\newcommand{\polone}{p}
\newcommand{\poltwo}{q}

\begin{document}
\maketitle

\begin{abstract}
We show that time complexity analysis of higher-order functional programs can be effectively
reduced to an arguably simpler (although computationally equivalent) verification problem, namely
checking first-order inequalities for validity. This is done by giving an efficient inference
algorithm for linear dependent types which, given a \PCF\ term, produces in output both a linear
dependent type and a cost expression for the term, together with a set of proof obligations. 
Actually, the output type judgement
is derivable \emph{iff} all proof obligations are valid. This, coupled with the already known
\emph{relative completeness} of linear dependent types, ensures that no information is lost, i.e.,
that there are no false positives or negatives. Moreover, the procedure reflects the difficulty
of the original problem: simple \PCF\ terms give rise to sets of
proof obligations which are easy to solve. The latter can then be put in a format suitable 
for automatic or semi-automatic verification by external solvers. Ongoing experimental evaluation 
has produced encouraging results, which are briefly presented in the paper.
\end{abstract}

\section{Introduction}
One of the most crucial non-functional properties of programs is the amount of resources (like time, memory
and power) they need when executed. Deriving upper bounds on the resource consumption of programs is
crucial in many cases, but is in fact an undecidable problem as soon as the underlying programming language is non-trivial.
If the units of measurement in which resources are measured become concrete and close to the physical
ones, the problem becomes even more complicated, given the many transformation and optimisation layers
programs are applied to before being executed. A typical example is the one of WCET techniques adopted
in real-time systems~\cite{wcet}, which do not only need to deal with how many machine instructions a program corresponds
to, but also with how much time each instruction costs when executed by possibly complex architectures
(including caches, pipelining, etc.), a task which is becoming even harder with the current trend
towards multicores architectures.

A different approach consists in analysing the \emph{abstract} complexity of programs. As an example,
one can take the number of instructions executed by the program as a measure of its execution
time. This is of course a less informative metric, which however becomes more accurate if the actual 
time complexity \emph{of each instruction} is kept low. One advantage of this analysis is the independence 
from the specific hardware platform executing the program at hand: the latter only needs
to be analysed once. A variety of verification techniques have been employed in this context, from
abstract interpretation~\cite{speed} to type systems~\cite{HOAA10} to program logics~\cite{asdf} to interactive theorem proving\footnote{A detailed discussion with related work is in Section~\ref{sec:rw}}.

Among the many type-based techniques for complexity analysis, a recent proposal consists in going towards
systems of \emph{linear dependent types}, as suggested by Marco Gaboardi and the first author~\cite{DLG11}. In linear
dependent type theories, a judgement has the form $\vdash_\I \t:`s$, where $`s$ is the type of $\t$ and
$\I$ is its \emph{cost}, an estimation of its time complexity. In this 
paper, we show that the problem of checking, given a \PCF\ term $\t$ and $\I$, whether
$\vdash_\I \t:`s$ holds can be efficiently reduced to the one of checking the truth of a set of proof obligations, themselves
formulated in the language of a \emph{first-order} equational program. Interestingly,
simple $\lambda$-terms give rise to simple equational programs. In other words, 
linear dependent types are not only a sound and \emph{relatively} complete methodology for inferring \emph{time} 
bounds of programs~\cite{DLG11,DLP12}: they also allow to reduce complexity analysis to an arguably simpler (although
computationally equivalent) problem which is much better studied and for which a variety of techniques and
concrete tools exist~\cite{why3}. Noticeably, the bounds one obtains this way translate to bounds on the number 
of steps performed by evaluation machines for the $\lambda$-calculus, which means that the induced metrics are not too 
abstract after all. The type inference algorithm is described in Section~\ref{sec:typeinf}.

The scenario, then, becomes similar to the one in Floyd-Hoare program
logics for imperative programs, where completeness holds~\cite{Cook78}
(at least for the simplest idioms~\cite{Clarke79}) \emph{and} weakest
preconditions can be generated automatically (see,
e.g.,~\cite{asdf}). A benefit of working with functional programs is that
type inference --- the analogue of generating WPs --- can be done
compositionally without the need of guessing invariants. 

Linear dependent types are simple types annotated with some 
\emph{index terms}, i.e. first-order terms reflecting the value of data flowing inside the program.
Type inference produces in output a type derivation, a set of inequalities (which should be thought
of as proof obligations) and an equational program  $\mathcal{E}$ giving meaning to function symbols appearing in index 
terms (see Figure~\ref{fig:context}). A natural thing to do once $\mathcal{E}$ and the various proof
obligations are available is to try to solve them automatically, as an example through SMT solvers.
If automatically checking the inequalities for truth does not succeed (which must happen, in some cases), one 
can anyway find useful information in the type derivation, as it tells you precisely \emph{which} data every symbol 
corresponds to. We elaborate on this issue in Section~\ref{sec:side_cond}.
\begin{figure}
  \centering
  \includegraphics[scale=1.4]{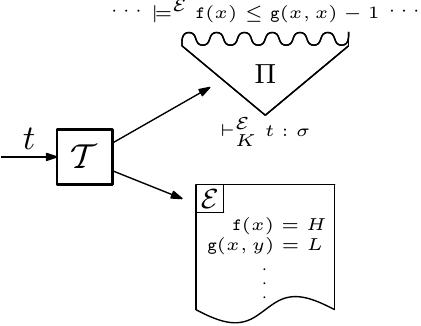}
  \caption{General scheme of the type inference algorithm.}
  \label{fig:context}
\end{figure}

But where does linear dependency come from? Linear dependent types can be seen as a way to turn Girard's geometry of 
interaction~\cite{Girard87} (or, equivalently, AJM games~\cite{AJM00}) into a type system for the $\lambda$-calculus:
the equational program one obtains as a result of type inference of a term $\t$ is nothing but as a description of a token machine
for $\t$. In presence of linear dependency, any term which can possibly be duplicated, can receive
different, although uniform, types, similarly to what happens in \BLL~\cite{GSS92}. 
As such, this form of dependency is significantly simpler than the one of, e.g., the calculus of inductive constructions. 


\section{Linear Dependency at a Glance}
\label{sec:informal}
Traditionally, type systems carry very little information about the \emph{value} of data manipulated by programs,
instead focusing on their \emph{nature}. As an example, all (partial recursive) functions from natural numbers to natural numbers
can be typed as $\NatPCF\toPCF\NatPCF$ in the $\lambda$-calculus with natural numbers and higher-order recursion, also known as \PCF~\cite{Plotkin77}.
This is not an intrinsic limit of the type-based analysis of programs, however: much richer type disciplines have flourished in the last 
twenty years~\cite{HughesPS96,Denney98,CraryWei00,BartheGR08}. All of them guarantee stronger 
properties for typable programs, the price being a more complicated type language and computationally more difficult type inference
and checking problems. As an example, sized types~\cite{HughesPS96} are a way to ensure termination of functional programs based on size
information. In systems of sized types, a program like
$$
\t=\lambda x.\lambda y.\add\ (\add\ x\ y)\ (\fsucc\ y)
$$
can be typed as $\NatPCF_{a}\toPCF\NatPCF_{b}\toPCF\NatPCF_{a+2b+\one}$, and in general as 
$\NatPCF_{a}\toPCF\NatPCF_{b}\toPCF\NatPCF_{\I}$, where $\I\geq a+2b+\one$.
In other words, the \PCF\ type $\NatPCF$ is refined into $\NatPCF_{\I}$ (where $\I$ is an arithmetical expression) whose semantics 
is the set of all natural numbers smaller or equal to $\I$, \ie\ the interval $[\zero,\I]\subseteq\int$.
The role of size information is to ensure that all functions terminate, and this is done by restricting the kind of functions 
of which one is allowed to form fixpoints. Sized types are nonlinear: arguments to functions can be freely duplicated.
Moreover, the size information is only approximate, since the expression labelling base types is only an \emph{upper} bound on the 
size of typable values.

Linear dependent types can be seen as a way to inject precision and linearity into sized types.
Indeed, $t$ receives the following type in \dlpcfv~\cite{DLP12}:
\begin{equation}\label{eq:extypeone}
  \Nat{\zero}{a}\marr{c<\I}\Nat{\zero}{b}\marr{d<\J}\Nat{\zero}{a+2b+\one}.
\end{equation}
As one can easily realise, $\Nat{\K}{\H}$ is the type of all natural numbers in the interval $[\K,\H]\subseteq\int$.
Moreover, $`s\marr{b<\J}`t$ is the type of \emph{linear} functions from $`s$ to $`t$ which can be copied by the environment $\J$ times. 
The $\J$ copies of the function have types obtained by substituting $\zero,\ldots,\J-\one$ for $b$ in $`s$ and $`t$.
This is the key idea behind linear dependency. The type~\eqref{eq:extypeone} is imprecise, but can be easily turned into
\begin{equation}\label{eq:extypetwo}
\Nat{a}{a}\marr{c<\I}\Nat{b}{b}\marr{d<\J}\Nat{a+2b+\one}{a+2b+\one},
\end{equation}
itself a type of $t$. In the following, the singleton interval type $\Nat{\K}{\K}$ is denoted simply as $\NatU{\K}$.

Notice that linear dependency is not exploited in~\eqref{eq:extypetwo}, \eg, $d$ does not appear free in $\NatU{b}$ nor in $\NatU{a+2b+\one}$.
Yet, \eqref{eq:extypetwo} precisely captures the functional behaviour of~\t. If $d$ does not appear free in $`s$ nor in $`t$, then
$`s\marr{d<\I}`t$ can be abbreviated as $`s\marr{\I}`t$. Linear dependency becomes necessary in presence of higher order functions. 
Consider, as another example, the term
$$
u=\lambda x.\lambda y.\ifz{y}{\nb[0]}{xy}
$$
$u$ has simple type $(\NatPCF\toPCF\NatPCF)\toPCF\NatPCF\toPCF\NatPCF$. One way to turn it into a linear dependent type is the following
\begin{equation}\label{eq:exttypethree}
(\NatU{a}\marr{\H}\NatU{\I})\marr{\K}\NatU{a}\marr{\L}\NatU{\J},
\end{equation}
where $\J$ equals $\zero$ when $a=\zero$ and $\J$ equals $\I$ otherwise.
Actually, $u$ has type (\ref{eq:exttypethree}) \emph{for every} $\I$ and $\J$, provided the two expressions are in the appropriate relation. Now, consider the term
$$
v=(\lambda x.\lambda y.(x\ \fpred\ (x\ \id\ y)))\ u.
$$
The same variable $x$ is applied to the identity $\id$ and to the
predecessor $\fpred$. Which type should we give to the variable $x$ and to $u$, then? 
If we want to preserve precision, the type should reflect \emph{both} uses of $x$. The right type for
$u$ is actually the following:
\begin{equation}\label{eq:typefour}
  (\NatU{a}\marr{\one}\NatU{\I})\marr{c<2}\NatU{a}\marr{\one}\NatU{\J}
\end{equation}
where both $\I$ and $\J$ evaluate to $a$ if $c=\zero$ and to $a-\one$ otherwise.
If $\id$ is replaced by $\fsucc$ in the definition of~\v, then~\eqref{eq:typefour} becomes even more complicated:
the first ``copy'' of $\J$ is fed not with $a$ but with either $\zero$ or $a+\one$.

Linear dependency precisely consists in allowing different copies of a term to receive types 
which are indexed differently (although having the same ``functional skeleton'')
and to represent all of them in compact form. This is in contrast to, \eg, intersection types, 
where the many different ways a function uses its argument could even be structurally different. 
This, as we will see in Section~\ref{sec:formal}, has important consequences on the kind of completeness 
results one can hope for: if the language in which \emph{index terms} 
are written is sufficiently rich, then the obtained system is complete in an intensional sense: 
a precise type can be given to \emph{every} terminating $t$ \emph{having type} $\NatPCF\toPCF\NatPCF$.

Noticeably, linear dependency allows to get precise information about the functional behaviour of programs without 
making the language of types too different from the one of simple types (\eg, one does not
need to quantify over index variables, as in sized types). The price to pay, however, is that types, and
especially higher-order types, need to be \emph{context aware}: when you type $\u$ as a subterm of $\v$ (see above) you
need to know which arguments $\u$ will be applied to. Despite this, a genuinely compositional type inference
procedure can actually be designed and is the main technical contribution of this paper.

\subsection{Linearity, Abstract Machines and the Complexity of Evaluation}
\label{sec:infor-lin}
Why dependency, but specially \emph{linearity}, are so useful for complexity analysis?
Actually, typing a term using linear dependent types requires finding an upper bound to the number of times each value is copied by its environment, called its \emph{potential}.
In the term~\v\ from the example above, the variable $x$ is used twice, and accordingly one finds $c<2$ in \eqref{eq:typefour}. 
Potentials of higher-order values occurring in a term are crucial parameters for the complexity of evaluating the term by abstract mechanisms~\cite{DalLago06}.
The following is an hopefully convincing (but necessarily informal) discussion about why this is the case.

Configurations of abstract machines for the $\lambda$-calculus (like Friedman and Felleisen's \CEK\ and Krivine's \KAM) can be thought of as 
being decomposable into two distinct parts:
\begin{varitemize}
\item
  First of all, there are \emph{duplicable} entities which are either copied entirely or turned into non-duplicable entities.
  This includes, in particular, terms in so-called environments.
  Each (higher-order) duplicable entity is a subterm of the term the computation started from.
\item
  There are \emph{non-duplicable} entities that the machine uses to look for the next redex to be fired.
  Typically, these entities are the current term and (possibly) the stack.
  The essential feature of non-duplicable entities is the fact that they are progressively ``consumed'' 
  by the machine during evaluation: the search for the next redex somehow consists in traversing 
  the non-duplicable entities until a redex is found \emph{or} a duplicable entity needs to be turned into a non-duplicable one.
\end{varitemize}
As an example, consider the process of evaluating the \PCF\ term
$$
(`lf.\ifz{(f\,\nb[0])}{\nb[0]}{f(f\,\nb[0])})((`lx.`ly.\add\,x\,y)~\nb[3])
$$
by an appropriate generalisation of the \CEK\ machine, see Figure~\ref{fig:exple-cek}.
\begin{figure*}[t]
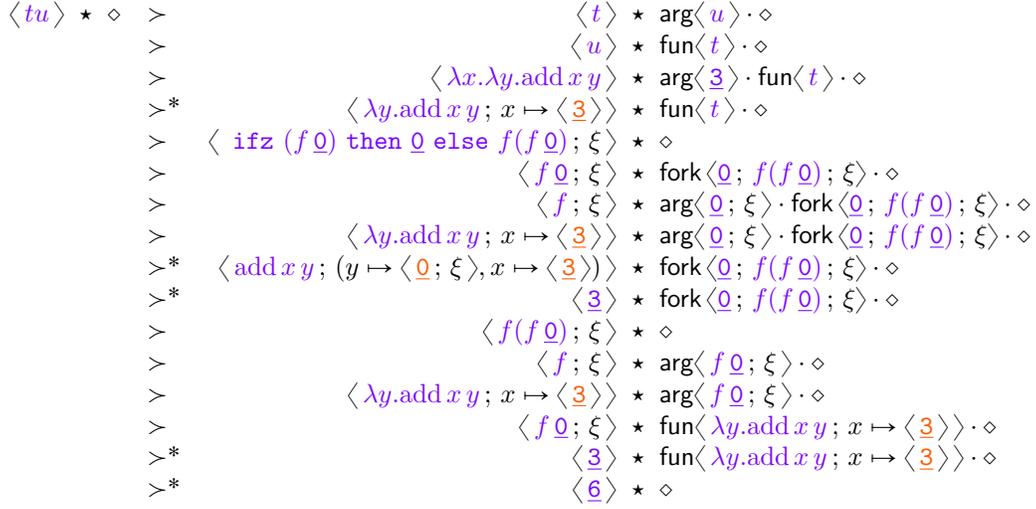

\begin{center}
$\t=`lf.\ifz{(f\,\nb[0])}{\nb[0]}{f(f\,\nb[0])};$
\quad\quad
$\u=(`lx.`ly.\add\,x\,y)~\nb[3];$
\quad\quad
$\env=\mkenv[f]{\clo[\mkenv{\cloee{\ctwo{\nb[3]}}}]{\ctwo{`ly.\add\,x\,y}}}.$
\end{center}
\begin{center}
  \begin{math}
    \begin{array}{clr@{~\star~}l}
      \proc{\cloee{\cone{\t\u}}}{\estk}
      & \tocek & \cloee{\cone{\t}} &
      \stkarg{\cloee{\cone{\u}}}{\estk}
      \\
      & \tocek & \cloee{\cone{\u}} 
      & \stkfun{\cloee{\cone{\t}}}{\estk}
      \\
      & \tocek & \cloee{\cone{`lx.`ly.\add\,x\,y}} & 
      \stkarg{\cloee{\cone{\nb[3]}}}{\stkfun{\cloee{\cone{\t}}}{\estk}}
      \\
      & \tocek^* 
      & \clo[\mkenv{\cloee{\ctwo{\nb[3]}}}]{\cone{`ly.\add\,x\,y}}
      & \stkfun{\cloee{\cone{\t}}}{\estk}
      \\
      & \tocek &
      \clo[\env]{\cone{\ifz{(f\,\nb[0])}{\nb[0]}{f(f\,\nb[0])}}}
      & \estk
      \\
      & \tocek & \clo[\env]{\cone{f\,\nb[0]}} &
      \stkif[\env]{\cone{\nb[0]}}{\cone{f(f\,\nb[0])}}{\estk}
      \\
      & \tocek & \clo[\env]{\cone{f}} &
      \stkarg{\clo[\env]{\cone{\nb[0]}}}{\stkif[\env]{\cone{\nb[0]}}{\cone{f(f\,\nb[0])}}{\estk}}
      \\
      & \tocek 
      & \clo[\mkenv{\cloee{\ctwo{\nb[3]}}}]{\cone{`ly.\add\,x\,y}} 
      &\stkarg{\clo[\env]{\cone{\nb[0]}}}{
        \stkif[\env]{\cone{\nb[0]}}{\cone{f(f\,\nb[0])}}{\estk}}
      \\
      & \tocek^* 
      &\clo[{(\mkenv[y]{\clo{\ctwo{\nb[0]}}},
        \mkenv{\cloee{\ctwo{\nb[3]}}})}]{\cone{\add\,x\,y}}
      &
      \stkif[\env]{\cone{\nb[0]}}{\cone{f(f\,\nb[0])}}{\estk}
      \\
      & \tocek^* & \cloee{\cone{\nb[3]}} &
      \stkif[\env]{\cone{\nb[0]}}{\cone{f(f\,\nb[0])}}{\estk}
      \\
      & \tocek & \clo[\env]{\cone{f(f\,\nb[0])}} & \estk
      \\
      & \tocek & \clo[\env]{\cone{f}} &
      \stkarg{\clo[\env]{\cone{f\,\nb[0]}}}{\estk}
      \\
      & \tocek 
      & \clo[\mkenv{\cloee{\ctwo{\nb[3]}}}]{\cone{`ly.\add\,x\,y}} 
      & \stkarg{\clo[\env]{\cone{f\,\nb[0]}}}{\estk}
      \\
      & \tocek & \clo[\env]{\cone{f\,\nb[0]}} & 
      \stkfun{\clo[\mkenv{\cloee{\ctwo{\nb[3]}}}]{\cone{`ly.\add\,x\,y}}}{\estk}
      \\
      & \tocek^* & \cloee{\cone{\nb[3]}} &
      \stkfun{\clo[\mkenv{\cloee{\ctwo{\nb[3]}}}]{\cone{`ly.\add\,x\,y}}}{\estk}
      \\
      & \tocek^* & \cloee{\cone{\nb[6]}} & \estk
    \end{array}
  \end{math}  
\end{center}
  \caption{Evaluation of a term in the \cek\ abstract machine.}
  \label{fig:exple-cek}
\end{figure*}
Initially, the whole term is non-duplicable. By travelling into it,
the machine finds a first redex $\u$; at that point, $\nb[3]$ becomes duplicable. 
The obtained closure itself becomes part of the environment~$\env$,
and the machine looks into the body of $\t$, ending up in an occurrence of $f$, which
needs to be replaced by a copy of $\env(f)$. After an \emph{instantiation step},
a new non-duplicable entity $\env(f)$ indeed appears. Note that, by an easy combinatorial argument, the number of machine steps necessary
to reach $f$ is at most (proportional to) the size of the starting term $\t\u$, since reaching
$f$ requires consuming non-duplicable entities which can only be created through
instantiations. After a copy of $\env(f)$ becomes non-duplicable, some additional ``nonduplicable fuel'' 
becomes available, but not too much: $`ly.\add\,x\,y$ is after all a subterm of the initial term. 

The careful reader should already have guessed the moral of this story: when analysing
the time complexity of evaluation, we could limit ourselves to counting how many
\emph{instantiation steps} the machine performs (as opposed to counting \emph{all}
machine steps). We claim, on the other hand, that the number of instantiation steps
\emph{equals} the sum of potentials of all values appearing in the initial
term, something that can be easily inferred from the kind of precise linear typing
we were talking about at the beginning of this section.

Summing up, once a dependently linear type has been attributed to a term $\t$, the time complexity
of evaluating $\t$ can be derived somehow for free: not only an expression bounding the number of instantiation
steps performed by an abstract machine evaluating $\t$ can be derived, but it is part of the underlying
type derivation, essentially. As a consequence, reasoning (automatically or not) about it can be
done following the structure of the program.

\section{Programs and Types, Formally}
\label{sec:formal}

In this section, we present the formalism of \dlpcf, a \emph{linear dependent type system} for~\PCF~\cite{Plotkin77}.
Two versions exist: \dlpcfn\ and \dlpcfv, corresponding to \emph{call-by-name} and \emph{call-by-value} evaluation of terms, respectively.
The two type systems are differerent, but the underlying idea is basically the same. We give here the details of the \cbv\ version~\cite{DLP12}, which
better corresponds to widespread intuitions about evaluation, but also provide some indications about the \cbn\ setting~\cite{DLG11}.

\subsection{Terms and Indexes}
\label{sec:term-ind}
Terms are given by the usual \PCF\ grammar:
\begin{align*}
  \s,\t,\u~:=&\;x\midd\nb\midd\t\,\u\midd`lx.\t\midd\pred(\t)\midd\suc(\t)\\
  & \fix{\t}\midd\ifz{\t}{\u}{\s}.
\end{align*}
A \emph{value} (denoted by~$v,w$ \etc) is either a primitive integer~\nb, or an abstraction~$`lx.\t$, or a fixpoint~\fix{\t}.
In addition to the usual terms of the $`l$-calculus, there are a fixpoint construction, primitive natural numbers with predecessor and successor, and 
conditional branching with a test for zero. For instance, a simple program computing addition is the following:
\begin{displaymath}
  \add~=~
  \fix{~`ly.`lz.\ifz{y}{z}{\suc(x~(\pred(y))~z)}}~.
\end{displaymath}

\subsubsection{Language of Indexes}
\label{sec:index}
As explained informally in Section~\ref{sec:informal}, a type in \dlpcf\ consists in an annotation of a \PCF\ type, where the
annotation consists in some \emph{indexes}. The latter are parametrised by a set of index variables~$\mathcal{V}=\{a,b,c,\dots\}$ 
and an untyped signature~$\Theta$ of \emph{function symbols}, denoted with~$\f$, $\g$, $\h$, \etc We assume~$\Theta$ contains at least the arithmetic 
symbols~$+$, \mnu, \zero\ and~\one, and we write \inb\ for $\one+\cdots+\one$ ($n$ times).
Indexes are then constructed by the following grammar:
\begin{displaymath}
  \I,\J,\K~::=\quad
  a\midd\f(\I_1,\ldots,\I_{n}) 
  \midd\sum_{a< \I}\J
  \midd\fc[a]{\I}{\J}{\K}~,
\end{displaymath}
where~$n$ is the arity of~\f\ in~$\Theta$.
Free and bound index variables are defined as usual, taking care that all occurrences of $a$ in $\J$ are bound in both~$\sum_{a< \I}\J$ and~$\fc[a]{\I}{\K}{\J}$.
The substitution of a variable~$a$ with~\J\ in~\I\ is written~$\I\isubst{\J}$.
Given an \emph{equational program}~\ep\ attributing a \emph{meaning} in $\int^n\rightharpoonup\int$ to some symbols of arity~$n$ 
(see Section~\ref{sec:side_cond}), and a \emph{valuation}~$`r$ mapping index variables to~\int, the \emph{semantics}~\itp{\I} 
of an index~\I\ is either a natural number or undefined. Let us describe how we interpret the last two 
constructions, namely \emph{bounded sums} and \emph{forest cardinalities}.

Bounded sums have the usual meaning: $\sum_{a<\I}\J$ is simply the sum of all possible values of $\J$ with $a$ taking the values from~\zero\ up to~$\I$, excluded. 
Describing the meaning of forest cardinalities, on the other hand, requires some effort.
Informally, the index $\fc[a]{\I}{\J}{\K}$ counts the number of nodes in a forest composed of $\J$ trees described using $\K$.
Each node in the forest is (uniquely) identified by a natural number, starting form~\I\ and visiting the tree in pre-order.
The index~$\K$ has the role of describing the number of children of each forest node, \eg\ the number of children of the node 
$\zero$ is $\K\isubst{\zero}$. \lv{More formally, the meaning of a forest cardinality is defined by   the following two equations:
\begin{align*}
  \fc[a]{\I}{0}{\K}&=0 \\
  \fc[a]{\I}{\J+1}{\K}&=
  \left(\fc[a]{\I}{\J}{\K}\right)+1+
  \left(\fc[a]{\I+1+\fc[a]{\I}{\J}{\K}}
       {\K\isubst{\I+\fc[a]{\I}{\J}{\K}}}{\K}\right)
\end{align*}
The first equation says that a forest of $0$ trees contains no nodes.
The second one tells us that a forest of $\J+1$ trees contains:
\begin{varitemize} 
\item
  The nodes in the first $\J$ trees;
\item
  plus the nodes in the last tree, which are just one plus the nodes   in the immediate 
  subtrees of the root, considered themselves as a forest.
\end{varitemize}
}
Consider the following forest comprising two trees:
$$
\scalebox{0.8}{
  \xymatrix{
    & &  0          &   \\
    & & 1\ar@{-}[u] &   \\
    \empty&2\ar@{-}[ur] &5\ar@{-}[u] & 6\ar@{-}[ul]\\
    3\ar@{-}[ur]& &  4\ar@{-}[ul]&7\ar@{-}[u]
    \save "4,1"."3,4"*[F.]\frm{--}\restore 
  }
  \xymatrix{
    & & 8 &   \\
    & 9\ar@{-}[ur] & & 11\ar@{-}[ul]\\
    & 10\ar@{-}[u] & & 12\ar@{-}[u] \\
  }}
$$
and consider an index~\K\ with a free index variable~$a$ such that 
$\K\isubst{1}=3$;
$\K\isubst{\inb}=2$ for $n\in\{2,8\}$;
$\K\isubst{\inb}=1$ when
$n\in\{0,6,9,11\}$;
and
$\K\isubst{\inb}=0$ when
$n\in\{3,4,7,10,12\}$. That is, $\K$ describes the number of children of each node. Then 
$\fc[a]{\zero}{2}{\K}=13$ since it takes into account the entire forest;
$\fc[a]{\zero}{1}{\K}=8$ since it takes into account only the leftmost tree;
$\fc[a]{8}{1}{\K}=5$ since it takes into account only the second tree of the forest;
finally, 
$\fc[a]{2}{3}{\K}=6$ since it takes into account only the three trees (as a forest) 
within the dashed rectangle.

One may wonder what is the role of forest cardinalities in the type system.
Actually, they play a crucial role in the treatment of recursion, where the 
unfolding of recursive calls produces a tree-like structure whose size is just the 
number of times the (recursively defined) function will be used \emph{globally}. 

Notice that $\itp{\I}$ is undefined whenever
the equality between $\I$ and any natural number cannot be derived from the underlying
equational program. In particular, a forest cardinality may be undefined even if all its subterms
are defined: as an example $\I=\fc[a]{0}{1}{1}$ has no value, because the 
corresponding tree consists of an infinite descending chain and its cardinality
is infinite. By the way $\I$ is the index term describing the structure of the recursive 
calls induced by the program $\fix{\t}$.

\subsubsection{Semantic Judgements}
\label{sec:sem-judg}
A \emph{constraint} is an inequality on indexes. A constraint $\I\leq\J$ is \emph{valid} 
for~$`r$ and~\ep\ when both $\itp{\I}$ and $\itp{\J}$ are defined, and $\itp{\I}\leq\itp{\J}$.
As usual, we can derive a notion of equality and strict inequality from~$\leq$. A 
\emph{semantic judgement} is of the form
\begin{displaymath}
  \ijudg{\fiv}{\ictx}{\I\leq\J}~,
\end{displaymath}
where~\ictx\ is a set of constraints and~\fiv\ is the set of free index variables in~\ictx, \I\ and~\J.
These semantic judgements are used as axioms in the typing derivations of~\dlpcf, and the set of 
constraints~\ictx, called the \emph{index context}, contains \emph{mainly} some indications of 
bounds for the free index variables (such as $a<\K$). Such a judgement is valid when, 
for every valuation $`r:\fiv→\int$, if all constraints in~\ictx\ are valid for~\ep\ and~$`r$ then so is $\I\leq\J$.

\subsection{Types}
\label{sec:types}
Remember that \dlpcf\ is aimed at controlling the complexity of programs.
The time complexity \textit{of the evaluation} is thus analysed statically, while typing the term at hand.
The grammar for types distinguishes the subclass of \emph{linear types}, which correspond to 
\textit{non-duplicable} terms (see Section~\ref{sec:infor-lin}), and the one of \emph{modal types}, for \emph{duplicable} terms.
In \dlpcfv, they are defined as follows:
\begin{displaymath}
  \begin{array}{rcl@{\qquad}l}
    \A, \B &::= & `s\toLL`t; &\text{linear types}\\
    `s,`t &::= & \mtyp{\I}{\A} \midd \Nat{\I}{\J}. &\text{modal types}
  \end{array}
\end{displaymath}
Indeed, \cbv\ evaluation only duplicates values. If such a value has an arrow type then it is a 
function (either an abstraction or a fixpoint) that can potentially increase the complexity of the whole program if we duplicate it.
Hence we may need a bound on the number of times we instantiate it in order to control the complexity.
This bound, call the \emph{potential} of the value, is represented by~\I\ in the type $\arrv{\I}{`s}{`t}$ (also written $`s\marr{a<\I}`t$).
As explained in Section~\ref{sec:informal}, \Nat{\I}{\J} is the type of programs evaluating to a natural number in the closed
interval $[\I,\J]$. The potential for natural numbers is not specified, as they can be freely duplicated along~\cbv\ evaluation.

\subsubsection{Summing types}
\label{sec:typ-sum}
Intuitively, the modal type $`s\equiv\mtyp{\I}{\A}$ is assigned to terms that can be copied~\I\ times, the $k^{th}$ 
copy being of type $\A\isubst{\inb[k]\mnu\one}$. For those readers who are familiar with Linear Logic, 
$`s$ can be thought of as representing the type $\A\isubst{\zero}\otimes\cdots\otimes\A\isubst{\I-\one}$.

In the typing rules we are going to define, modal types need to be manipulated in an algebraic way.
For this reason, two operations on modal types are required.
The first one is a binary operation~$\uplus$ on modal types.
Suppose that $`s\equiv\mtyp{\I}{\A\isubst[c]{a}}$
and that $`t\equiv\mtyp{\J}{\A\isubst[c]{\I+a}}$.
In other words, $`s$ consists of the first $\I$ instances of $\A$, \ie\ $\A\isubst[c]{\zero}\otimes\cdots\otimes\A\isubst[c]{\I-\one}$ 
while $`t$ consists of the next $\J$ instances of $\A$, \ie\ $\A\isubst[c]{\I+\zero}\otimes\cdots\otimes\A\isubst[c]{\I+\J-\one}$.
Their \emph{sum} $`s\uplus`t$ is naturally defined as a modal type consisting of the first $\I+\J$ instances of $\A$, \ie\ $\mtyp[c]{\I+\J}{\A}$.
Furthermore, $\Nat{\I}{\J}\uplus\Nat{\I}{\J}$ is just $\Nat{\I}{\J}$.
A bounded sum operator on modal types can be defined by generalising the idea above:
suppose that 
$$
`s=\mtyp[b]{\J}{\A\left\{b+\sum_{d<a}\J\isubst{d}/c\right\}}.
$$
Then its \emph{bounded sum} $\sum_{a<\I}`s$ is just $\mtyp[c]{\sum_{a<\I}\J}{\A}$.
Finally, $\sum_{a<\I}\Nat{\J}{\K}=\Nat{\J}{\K}$, provided $a$ does not occur free in $\J$ nor in~\K.

\subsubsection{Subtyping}
\label{sec:subtyp}

\begin{figure}                                            %
  \centering
  \begin{displaymath}
    \frac{
      \begin{array}{c}
        \ijudg{\fiv}{\ictx}{\K\leq\I} \\
        \ijudg{\fiv}{\ictx}{\J\leq\H} \\
      \end{array}
    }{\sjudg{\fiv}{\ictx}{\Nat{\I}{\J}`<\Nat{\K}{\H}}}
    \qquad
    \frac{
      \begin{array}{c}
        \sjudg{\fiv}{\ictx}{`u`<`s}\\
        \sjudg{\fiv}{\ictx}{`t`<`v}
      \end{array}
    }{\sjudg{\fiv}{\ictx}{`s\toLL`t`<`u\toLL`v}}
    \end{displaymath}
    \begin{displaymath}
      \frac{
        \begin{array}{r@{\,}l}
          \sjudg{(a,\fiv)}{(a<\J,\ictx)&}{\A`<\B}\\
          \ijudg{\fiv}{\ictx&}{\J\leq\I}
        \end{array}
      }{\sjudg{\fiv}{\ictx}{\mtyp{\I}{\A} `< \mtyp{\J}{\B}}}
  \end{displaymath}
  \caption{Subtyping derivation rules of \dlpcfv.}
  \label{fig:tsub}
\end{figure}                                              %

Central to \dlpcf\ is the notion of subtyping.
An inequality relation $`<$ between (linear or modal) types can be defined using the formal system in Fig.~\ref{fig:tsub}. 
This relation corresponds to lifting index inequalities to the type level.
As defined here,~$`<$ is a pre-order (\ie\ a reflexive and transitive relation), which allows to cope with approximations in the typed analysis of programs.
However, in the type inference algorithm we will present in next section only the symmetric closure $\equiv$ of $`<$, called \emph{type equivalence} will be used.
This ensures that the type produced by the algorithm is \emph{precise}.

\subsubsection{Typing}
\label{sec:typing}

\begin{figure*}                                            %
  \begin{displaymath} 
    \frac{}{\judg{\fiv}{\ictx}{`G,x:`s}{\zero}{x}{`s}}(\textit{Ax})
    \qquad
    \frac{\judg{\fiv}{\ictx}{`G}{\I}{t}{`s}
      \qquad
      \sjudg{\fiv}{\ictx}{`D`<`G}
      \qquad
      \sjudg{\fiv}{\ictx}{`s`<`t}
      \qquad
      \ijudg{\fiv}{\ictx}{\I\leq\J}
    }{\judg{\fiv}{\ictx}{`D}{\J}{t}{`t}}(\textit{Subs})
  \end{displaymath}
  \begin{displaymath}
    \frac{
      \judg{(a,\fiv)}{(a<\I,\ictx)}{`G,x:`s}{\K}{\t}{`t}
    }{
      \judg{\fiv}{\ictx}{\sum_{a<\I}`G}{\I+\sum_{a<\I}\K}{`lx.\t}{\mtyp{\I}{`s\toLL`t}}
    }(\toLL)
      \lvsv{\quad}{\qquad}
      \frac{
        \judg{\fiv}{\ictx}{`G}{\K}{\t}{\mtyp{\one}{`s\toLL`t}}
        \qquad 
        \judg{\fiv}{\ictx}{`D}{\H}{\u}{`s\isubst{\zero}}
      }{
        \judg{\fiv}{\ictx}{`G\uplus`D}{\K+\H}{\t\u}{`t\isubst{\zero}}
      }(\textit{App})
  \end{displaymath}   
  \begin{displaymath}
    \quad
     \frac{
       \judg{\fiv}{\ictx}{`G}{\M}{\t}{\Nat{\J}{\K}}
       \qquad 
       \judg{\fiv}{(\J\leq\zero,\ictx)}{`D}{\N}{\u}{`t}
       \qquad
       \judg{\fiv}{(\K\geq\one,\ictx)}{`D}{\N}{\s}{`t}
     }{
       \judg{\fiv}{\ictx}{`G\uplus`D}{\M+\N}{\ifz{\t}{\u}{\s}}{`t}
     }(\textit{If})
  \end{displaymath}
  \begin{displaymath}
    \frac{}{\judg{\fiv}{\ictx}{`G}{\zero}{\nb}{\Nat{\inb}{\inb}}}(n)
    \qquad
    \frac{
      \judg{\fiv}{\ictx}{`G}{\M}{\t}{\Nat{\I}{\J}}
    }{
      \judg{\fiv}{\ictx}{`G}{\M}{\suc(\t)}{\Nat{\I+\one}{\J+\one}}
    }(s)
      \qquad
      \frac{
        \judg{\fiv}{\ictx}{`G}{\M}{\t}{\Nat{\I}{\J}}
      }{
        \judg{\fiv}{\ictx}{`G}{\M}{\pred(\t)}{\Nat{\I\mnu\one}{\J\mnu\one}}
      }(p)
  \end{displaymath}
\lvsv{
  \begin{displaymath}
    \frac{
      \begin{array}{r@{\,}l}
        \judg{(b,\fiv)}{(b<\H,\ictx)}{
          `G,x:\mtyp{\I}{\A}&}{\J}{\t}{\mtyp{1}{\B}}
        \\
        \sjudg{(a,b,\fiv)}{(a<\I,b<\H,\ictx)&}{
          \B\isubst{\zero}\isubst[b]{\fc{b+1}{a}{\I}+b+1}`<\A}
      \end{array}
    }{\judg{\fiv}{\ictx}{\sum_{b<\H}`G}{
        \H+\sum_{b<\H}\J}{\fix{\t}}{\mtyp{\K}{
          \B\isubst{\zero}\isubst[b]{\fc{0}{a}{\I}}}}}
    (\textit{Fix})
  \end{displaymath}
 }{
   \begin{displaymath}
     \frac{
      \judg{(b,\fiv)}{(b<\H,\ictx)}{
        `G,x:\mtyp{\I}{\A}}{\J}{\t}{\mtyp{1}{\B}}
      \qquad
      \sjudg{(a,b,\fiv)}{(a<\I,b<\H,\ictx)}{
        \B\isubst{\zero}\isubst[b]{\fc{b+1}{a}{\I}+b+1}`<\A}
    }{\judg{\fiv}{\ictx}{\sum_{b<\H}`G}{
        \H+\sum_{b<\H}\J}{\fix{\t}}{\mtyp{\K}{
          \B\isubst{\zero}\isubst[b]{\fc{0}{a}{\I}}}}}
    (\textit{Fix})
    \end{displaymath}
}
 \vspace{-10pt}

    \hfill{\small (where $\H=\fc{0}{\K}{\I}$)}
  \caption{Typing rules of \dlpcfv.}
  \label{fig:typ}
\end{figure*}                                              %

A typing judgement is of the form
\begin{displaymath}
  \judg{\fiv}{\ictx}{`G}{\K}{\t}{`t},
\end{displaymath}
where~\K\ is the \emph{weight} of~\t, that is (informally) the maximal number of 
substitutions involved in the \cbv\ evaluation of~\t\ (including the potential substitutions by~\t\ itself in its evaluation context).
The index context~\ictx\ is as in a semantic judgement (see Section~\ref{sec:sem-judg}), and~$`G$ is a (term) context assigning a 
modal type to (at least) each free variable of~\t. Both sums and bounded sums are naturally extended from modal types to contexts 
(with, for instance, $\{x:`s;y:`t\}\uplus\{x:`u,z:`v\}=\{x:`s\uplus`u;y:`t;z:`v\}$).
There might be free index variables in~$\ictx,`G,`t$ and~\K, all of them from~\fiv.
Typing judgements can be derived from the rules of Figure~\ref{fig:typ}.

Observe that, in the typing rule for the abstraction~($\toLL$), \I\ represents the number of times the value~$`lx.\t$ can be copied.
Its weight (that is, the number of substitutions involving $`lx.\t$ or one of its subterms) is then~\I\ plus, for each of 
these copies, the weight of~\t. In the typing rule~(\textit{App}), on the other hand, \t\ is used once as a function, without been copied.
Its potential needs to be at least~\one. The typing rule for the fixpoint is arguably the most complicated one.
As a first approximation, assume that only one copy of~\fix{t} will be used (that is, $\K\equiv\one$ $a$ does not occur free in $B$).
To compute the weight of \fix{\t}, we need to know the number of times~\t\ will be copied during the evaluation, 
that is the number of nodes in the tree of its recursive calls. This tree is described by the index~\I\ (as explained 
in Section~\ref{sec:index}), since each occurrence of~$x$ in~\t\ stands for a recursive call. It has~$\H=\fc{\zero}{\one}{\I}$ nodes.
At each node~$b$ of this tree, there is a copy of~\t\ in which the $a^{th}$ occurrence of~$x$ will be replaced by the $a^{th}$ 
son of~$b$, \ie\ by $b+\one+\fc{b+\one}{a}{\I}$. The types thus have to correspond, which is what the second premise of 
this rule prescribes. Now if \fix{\t} is in fact aimed at being copied $\K\geq\zero$ times, then all the copies of~\t\ are 
represented in a forest of~\K\ trees described by~\I. 

For the sake of simplicity, we present here the type system with an explicit subsumption rule.
The latter allows to relax any bound in the types (and the weight), thereby loosing some precision in the information 
provided by the typing judgement. However, we could alternatively replace this rule by relaxing the premises of all 
the other ones (which corresponds to the presentation of the type system given in~\cite{DLP12}, or in~\cite{DLG11} for \dlpcfn).
Restricting subtyping to type equivalence amounts to considering types up to index equality in 
the type system of Figure~\ref{fig:typ} without the rule (\textit{Subs}) --- this is what we do in the type inference algorithm 
in Section~\ref{sec:typeinf}. In this case we say that the typing judgements are \emph{precise}:
\begin{definition}
  \label{def:precise}
  A derivable judgement~\judg{\fiv}{\!\ictx}{`G\!}{\I}{\t}{`s} is \emph{precise} if
  \begin{displaymath}
    \judg{\fiv}{\ictx}{`D}{\J}{\t}{`t} \text{ is derivable }
    \quad\implies\quad
    \left\{
      \begin{array}[c]{r@{\ }l@{\,}l}
        \sjudg{\fiv}{\ictx&}{`D&`<`G} \\
        \sjudg{\fiv}{\ictx&}{`s&`<`t} \\
        \ijudg{\fiv}{\ictx&}{\I&\leq\J} 
      \end{array}
    \right.
  \end{displaymath}
\end{definition}

\subsubsection{Call-by-value vs. Call-by-name}
\label{sec:typ-cbn-cbs}

In \dlpcfn, the syntax of terms and of indexes is the same as in \dlpcfv, but the language of types differs:
\begin{displaymath}
  \begin{array}{rcl@{\qquad}l}
    \A, \B &::= & `s\toLL\A \midd \Nat{\I}{\J}; &\text{linear types}\\
    `s,`t &::= & \mtyp{\I}{\A}. &\text{modal types}
  \end{array}
\end{displaymath}
Modal types still represent duplicable terms, except that now not only values but any argument to functions
can be duplicated. So modal types only occur in negative position in arrow types.
In the same way, one can find them in the \textit{context} of any typing judgement,
\begin{displaymath}
  \judg{\fiv}{\ictx}{(x_1:`s_1,\dots,x_n:`s_n)}{\K}{\t}{A}.
\end{displaymath}
When a term is typed, it is \textit{a priori} not duplicable, and its type is linear.
It is turned into a duplicable term when it holds the position of the argument in an application.
As a consequence, the typing rule (\textit{App}) becomes the most ``expansive'' one (for the weight) 
in \dlpcfn: the whole context used to type the argument has to be duplicated, whereas in \dlpcfv\ this duplication 
of context is ``anticipated'' in the typing rules for values.

The readers who are familiar with linear logic, could have noted that, if we replace modal types by banged types, 
and we remove all annotations with indexes, then \dlpcfn\ corresponds to the target fragment of the \cbn\ translation 
from simply-typed $`l$-calculus to \LL, and \dlpcfv\ to the target of the \cbv\ translation~\cite{MaraistOTW95}.

In \dlpcfn, the weight~\K\ of a typing judgement represents the maximal number of substitutions that may occur in the 
\cbn\ evaluation of~\t. We do not detail the typing rules of \dlpcfn\ here (they can be found in~\cite{DLG11}).
However an important remark is that in \dlpcfn, just like in \dlpcfv, some semantic judgements can be found 
in the axioms of a typing derivation, and that every typing rule is reversible (except subsumption).
The type inference algorithm for \dlpcfv\ that we present in Section~\ref{sec:typeinf} can be easily adapted to \dlpcfn.

\subsection{Abstract Machines}
\label{sec:am}
The evaluation of \PCF\ terms can be simulated through an extension~\kam\ of Krivine's abstract machine~\cite{Krivine07} 
(for \cbn\ evaluation) or through an extension~\cek\ of the Felleisen and Friedman's \CEK\ machine~\cite{FelleisenF87} (for 
\cbv\ evaluation).

Both these machines have states in the form of \emph{processes}, that are pairs of a \emph{closure} (\ie\, a term with 
an environment defining its free variables) and a \emph{stack}, representing the evaluation context. In the \kam, 
these objects are given by the following grammar:
\begin{displaymath}
  \begin{array}{l@{\qquad}rcl}
    \text{Closures:} & \cloc &:=& \clo{\t};\\
    \text{Environment:} & \env &:=&
    \{\mkenv[x_1]{\cloc_1};\cdots;\mkenv[x_k]{\cloc_k}\};\\
    \text{Stacks:} & \stk &:=& \estk\midd\stkarg{\clo{\t}}{\stk}
    \midd\stks{\stk}\midd\stkp{\stk}\\
    &&& \quad\midd\stkif{\t}{\u}{\stk};\\
    \text{Processes:} & P &:=& \proc{\cloc}{\stk}.
  \end{array} 
\end{displaymath}
When the environment is empty, we may use the notation~\cloee{\t} instead of~\clo[∅]{\t} for closures.
The evaluation rules of the \kam\ are given in Figure~\ref{fig:kam}. The fourth evaluation rule is said
to be an \emph{instantiation step}: the value of a variable $x$ is replaced by the term $x$ maps to in the
underlying environment $\xi$. 

The \cek\ machine, which performs \cbv\ evaluation, is slightly more complex:
within closures, the \emph{value closures} are those whose first component is
a value:
\begin{math}
  \clov~:=~\clo{\v}
\end{math}
\
(remember that a \emph{value}~\v\ is of the form~\nb, $`lx.\t$ or~\fix{\t}).
Moreover, environments assign only value closures to variables:
\begin{displaymath}
  \env~:=\quad
  \{\mkenv[x_1]{\clov_1};\cdots;\mkenv[x_k]{\clov_k}\}~.
\end{displaymath}
The grammar for stacks is the same, with one additional construction
(\stkfun{(\clov)}{\stk}) that is used to encapsulate a function (lambda abstraction or fixpoint) while its argument is computed.
Indeed, the latter cannot be substituted for a variable if it is not a value.
The evaluation rules for processes are the same than in Figure~\ref{fig:kam}, except 
that the second and the third ones are replaced by the following:
\begin{displaymath}
  \begin{array}{c@{~\star~}c@{~~\tocek~~}c@{~\star~}c}
    \clov & \stkarg{(\cloc)}{\stk} & 
    \cloc & \stkfun{(\clov)}{\stk} 
    \\
    \clov & \stkfun{\clo{`lx.\t}}{\stk} &
    \clo[\mkenv{\clov}\cdot\env]{\t} & \stk \\
    \lvsv{
      \clov & \stkfun{\clo{\fix{\t}}}{\stk} & 
      \clo[\mkenv{\clo{\fix{\t}}}\cdot\env]{\t}
      & \stkarg{(\clov)}{\stk} \\
    }{
      \clov & \stkfun{\clo{\fix{\t}}}{\stk} & 
      \multicolumn{2}{c}{}\\
      \multicolumn{4}{r}{\proc{
          \clo[\mkenv{\clo{\fix{\t}}}\cdot\env]{\t}
        }{
          \stkarg{(\clov)}{\stk}}} \\
    }
      
  \end{array}
\end{displaymath}
An example of the evaluation of a term in the \cek\ machine has been shown in Figure~\ref{fig:exple-cek}.

We say that a term~\t\ \emph{evaluates} to~\u\ in an abstract machine when $\proc{\cloee{\t}}{\estk}\tocek\proc{\clo{\u}}{\estk}$.
Observe that if~\t\ is a closed term then~\u\ is necessarily a value. We write $\t\nconv\u$ whenever the \kam\ evaluates~\t\ to~\u\ in exactly~$n$ 
steps, and $\t\vconv\u$ when the same holds for the \cek\ (we may also omit the exponent~$n$ when the 
number of steps is not relevant).

\begin{figure*}                                            %
  \centering
  \begin{math}
    \begin{array}{c@{\quad\star\quad}c@{\qquad\tocek\qquad}c@{\quad\star\quad}c}
      \clo{\t\u} & \stk & \clo{\t} & \stkarg{\clo{\u}}{\stk} \\
      \clo{`lx.\t} & \stkarg{(\cloc)}{\stk} & 
      \clo[(\mkenv{\cloc})`.\env]{\t} & \stk 
      \\
      \clo{\fix{\t}} & \stk &
      \clo[(\mkenv{\clo{\fix{\t}}})`.\env]{\t} & \stk 
      \\
      \clo{x} & \stk & \env(x) & \stk \\
      \clo{\ifz{\t}{\u}{\s}} & \stk & \clo{\t} & \stkif{\u}{\s}{\stk} \\
      \clo[\env']{\nb[0]} & \stkif{\t}{\u}{\stk} & \clo{\t} &\stk\\
      \clo[\env']{\nb[n\!+\!1]} & \stkif{\t}{\u}{\stk} & \clo{\u} & \stk 
      \\
      \clo{\suc(t)} & \stk & \clo{\t} & \stks{\stk} \\
      \clo{\pred(t)} & \stk & \clo{\t} & \stkp{\stk} \\
      \clo{\nb} & \stks{\stk} & \cloee{\nb[n\!+\!1]} & \stk \\
      \clo{\nb} & \stkp{\stk} & \cloee{\nb[n\!-\!1]} & \stk \\
    \end{array}
  \end{math}
  \caption{\kam\ evaluation rules.}
  \label{fig:kam}
\end{figure*}                                              %

\paragraph{Abstract Machines and Weight}

The \emph{weight} of a typable term was informally presented as the number of instantiation steps in its evaluation.
Abstract machines enable a more precise formulation of this idea:
\begin{fact}
  \begin{varenumerate}
  \item If $\t\nconv[]\u$, and \ejudg{}{\I}{\t}{\A} is derivable in \dlpcfn, then~\itp[]{\I} is an upper bound for the number of instantiation steps 
    in the evaluation of~\t\ by the \kam.
  \item If $\t\vconv[]\u$, and \ejudg{}{\I}{\t}{\mtyp{\one}{\A}} is derivable in \dlpcfv, then~\itp[]{\I} is an upper bound for the instantiation steps
    %
      in the evaluation of~\t\ in the \cek.
\end{varenumerate}
\end{fact}
This can be shown by extending the notion of weight and of typing judgement to stacks and processes~\cite{DLG11,DLP12}, and is the main
ingredients for proving Intensional Soundness (see Section~\ref{sec:prop-dlpcf}).

\subsection{Key Properties}
\label{sec:prop-dlpcf}

In this section we briefly recall the main properties of~\dlpcf, arguing for its relevance as a methodology for complexity analysis.
We give the results for \dlpcfv, but they also hold for \dlpcfn\ (all proofs can be found in \cite{DLG11,DLP12}).

The subject reduction property guaranties as usual that typing is correct with respect to term evaluation, but specifies also that the 
weight of a term cannot increase along reduction:
\begin{proposition}[Subject Reduction]
  \label{prop:sr}
  For any \PCF-terms \t,\u, if $\judg{\fiv}{\ictx}{∅}{\I}{\t}{`t}$ is derivable in \dlpcfv, and 
  if $\t→\u$ in \cbv, then \judg{\fiv}{\ictx}{∅}{\J}{\u}{`t} is also derivable for some \J\ such that \ijudg{\fiv}{\ictx}{\J\leq\I}~.
\end{proposition}
As a consequence, the weight does not tell us much about the number of reduction steps bringing a (typable) term to its
normal form. So-called \emph{Intensional Soundness}, on the other hand, allows to deduce some sensible information about 
the time complexity of evaluating a typable \PCF\ program \emph{in an abstract machine} from its \dlpcf\ typing judgement.
\begin{proposition}[Intensional Soundness]
  \label{prop:int-snd}
  For any term~\t, if $\ejudg{}{\K}{\t}{\Nat{\I}{\J}}$ is derivable in \dlpcfv, then~\t\ evaluates to~\nb\ in~$k$ 
  steps in the~\cek, with $\itp[]{\I}\leq n\leq\itp[]{\J}$ and $k\leq|\t|`.(\itp[]{\K}+1)$~.
\end{proposition}
Intensional Soundness guarantees that the evaluation of any program typable in \dlpcf\ takes (at most)
a number of steps directly proportional to both its syntactic size and its weight. A similar theorem holds 
when $\t$ has a functional type: if, as an example, the type of $\t$ is $\NatU{a}\marr{1}\NatU{\J}$, then
$\K$ is parametric on $a$ and $(|\t|+2)`.(\itp[]{\K}+1)$ is an upper bound on the complexity of evaluating $\t$ 
on any integer $a$.

But is \dlpcf\ powerful enough to type natural complexity bounded programs?
Actually, it is as powerful as \PCF\ itself, since any \PCF\ type derivation can be turned into a \dlpcf\ one 
(for an expressive enough equational program), as formalised by the type inference 
algorithm (Section~\ref{sec:typeinf}). We can make this statement even more precise for terms of base or
first order type, provided two conditions are satisfied:
\begin{varitemize}
\item
  On the one hand, the equational program~\ep\ needs to be   \emph{universal}, meaning that every partial recursive 
  function is representable by some index term. This can be guaranteed, as an example, by the presence of a universal program in~\ep.
\item
  On the other hand, all \emph{true} statements in the form \ijudg{\fiv}{\ictx}{\I\leq\J} must be ``available'' in the type system for completeness to hold.
  In other words, one cannot assume that those judgements are derived in a given (recursively enumerable) formal system, because this would violate G\"odel's Incompleteness Theorem.
  In fact, in \dlpcf\ completeness theorems are \emph{relative} to an oracle for the truth of those assumptions, which is precisely what happens in Floyd-Hoare logics~\cite{Cook78}.
\end{varitemize}
\begin{proposition}[Relative Completeness]
  \label{prop:rel-cpl}
  If\/ \ep\ is universal, then for any \PCF\ term~\t,
  \begin{varenumerate}
  \item
    if $\t\vconv[k]\nb[m]$, then \ejudg{}{\inb[k]}{\t}{\NatU{\inb[m]}} is derivable in \dlpcfv.
  \item 
    if, for any $n`:\int$, there exist $k_n,m_n$ such that $\t\,\nb\vconv[k_n]\nb[m_n]$, then there exist~\I\ and~\J\ such that
    $$\judg{a}{∅}{∅}{\I}{\t}{\arrv[b]{\one}{\NatU{a}}{\NatU{\J}}}$$
    is derivable in \dlpcfv, with \eijudg{\J\isubst{\inb}=\inb[m_n]} and \eijudg{\I\isubst{\inb}\leq\inb[k_n]} for all~$n`:\int$.
  \end{varenumerate}
\end{proposition}
The careful reader should have noticed that there is indeed a gap between the lower bound
provided by completeness and the upper bound provided by soundness: this is indeed the reason
why our complexity analysis is only meaningful in an asymptotic sense. Sometimes, however,
programs with the same asymptotic behavior \emph{can} indeed be distinguished, e.g. when their
size is small relative to the constants in their weight.

In the next section, we will see how to make a concrete use of Relative Completeness.
Indeed, we will describe an algorithm that, given a \PCF\ term, returns a \dlpcf\ judgement 
\ejudg{}{\K}{\t}{`t} for this term, where $\ep$ is equational program that is not \textit{universal}, 
but expressive enough to derive the typing judgement. To cope with the ``relative'' part of the 
result (\ie, the very strong assumption that every \textit{true} semantic judgement must be available), 
the algorithm also return a set of \emph{side conditions} that have to be checked. These side 
conditions are in fact semantic judgements that act as axioms (of instances of the subsumption rule) in the 
typing derivation.


\section{Relative Type Inference}
\label{sec:typeinf}
Given on the one hand soundness and relative completeness of \dlpcf, and on the other undecidability 
of complexity analysis for \PCF\ programs, one may wonder whether looking for a type inference procedure makes sense at all.
As stressed in the Introduction, we will not give a type inference algorithm \emph{per se}, but rather reduce 
type inference to the problem of checking a set of inequalities modulo an equational program (see Figure~\ref{fig:context}).
This is the reason why we can only claim type inference to be algorithmically solvable in a \emph{relative} sense, \ie\  assuming the
existence of an oracle for proof obligations. 

Why is solving relative type inference useful? Suppose you have a program $\t:\NatPCF\toPCF\NatPCF$
and you want to prove that it works in a number of steps bounded by a polynomial
$p:\int\rightarrow\int$ (\eg, $p(x)=4\cdot x+7$). You could of course proceed by building a \dlpcf\ type
derivation for $\t$ by hand, or even reason directly on the complexity of $\t$.
Relative type inference simplifies your life: 
it outputs an equational program $\ep$, a \emph{precise} type derivation for $\t$ whose conclusion is 
$\judg{a}{\emptyset}{\emptyset}{\I}{\t}{\NatU{a}\marr{1}\NatU{\J}}$ and a set $\mathcal{I}$ of 
inequalities on the same signature as the one of $\ep$. Your original problem, then, is reduced to verifying 
$\models^\ep\mathcal{I}\cup\{\I\leq p(a)\}$. This is arguably an easier problem
than the original one: first of all, it has nothing to do with complexity analysis but is rather a problem
about the \emph{value} of arithmetical expressions. Secondly it only deals with first-order expressions.

\subsection{An Informal Account}
From the brief discussion in Section~\ref{sec:informal}, it should be clear that devising a 
compositional type inference procedure for \dlpcf\ is nontrivial: the type one assigns to a subterm
heavily depends on the ways the rest of the program uses the subterm. The solution we
adopt here consists in allowing the algorithm to return \emph{partially unspecified} equational programs:
$\ep$ as produced in output by $\mathcal{T}$ gives meaning to all the symbols in
the output type derivation \emph{except} those occurring in negative position in its conclusion.

To better understand how the type inference algorithm works, let us consider the following term $\t$:
$$
\u\v=(`lx.`ly.x(xy))(`lz.\suc(z)).
$$
The subterm $\u$ can be given type $(\NatPCF\toPCF\NatPCF)\toPCF\NatPCF\toPCF\NatPCF$ in \PCF, while
$\v$ has type $\NatPCF\toPCF\NatPCF$. This means $\t$ as a whole has type $\NatPCF\toPCF\NatPCF$ 
and computes the function $x\mapsto 2\cdot x$. The type inference algorithm proceeds by giving types to $\u$ and
to $\v$ separately, then assembling the two into one. Suppose we start with $\v$. The type inference
algorithm refines $\NatPCF\toPCF\NatPCF$ into $`s=\NatU{\f(a,b)}\marr{b<\h(a)}\NatU{\g(a,b)}$ and the
equational program $\rep_\v$, which gives meaning to $\g$ in terms of $\f$:
\begin{align*}
\g(a,b)&=\f(a,b)+1.
\end{align*}
Observe how both $\f$ and $\h$ are not specified in $\rep_\v$, because they appear in \emph{negative} position in $`s$:
$\f(a,b)$ intuitively corresponds to the argument(s) $\v$ will be applied to, while $\h(a)$ is the number of times
$\v$ will be used. Notice that everything is parametrised on $a$, which is something like a global parameter
that will later be set as the input to $\t$. The function $\u$, on the other hand, is given type
$$
\begin{array}{c}
(\NatU{\p(a,b,c)}\marr{c<\j(a,b)}\NatU{\q(a,b,c)})\\
\marr{b<\k(a)}\\
\NatU{\l(a,b,c)}\marr{c<\m(a,b)}\NatU{\n(a,b,c)}.
\end{array}
$$
The newly introduced function symbols are subject to the following equations:
\begin{align*}
\j(a,b)&=2\cdot \m(a,b);\\
\n(a,b,c)&=\q(a,b,2c);\\
\p(a,b,2c)&=\q(a,b,2c+1);\\
\p(a,b,2c+1)&=\l(a,b,2c).
\end{align*}
Again, notice that some functions are left unspecified, namely $\l$, $\m$, $\q$ and $\k$.
Now, a type for $\u\v$ can be found by just combining the types for $\u$ and $\v$,
somehow following the typing rule for applications. First of all, the number of times
$\u$ needs to be copied is set to $1$ by the equation $\k(a)=1$. Then, the matching
symbols of $\u$ and $\v$ are defined one in terms of the others:
\begin{align*}
\q(a,0,b)&=\g(a,b);\\
\f(a,b)&=\p(a,0,b);\\
\h(a)&=\j(a,0).
\end{align*}
This is the last step of type inference, so it is safe to stipulate that
$\m(a,0)=1$ and that $\l(a,0,c)=a$, thus obtaining a fully specified equational
program $\ep$ and the following type $`t$ for $\t$:
$$
\NatU{a}\marr{c<1}\NatU{\n(a,0,c)}.
$$
As an exercise, the reader can verify that the equational program above allows to
verify that $\n(a,0,0)=a+2$, and that 
$$
\judg{a}{\emptyset}{\emptyset}{2}{\t}{`t}.
$$

\subsection{Preliminaries}

Before embarking on the description of the type inference algorithms, some preliminary concepts and ideas need to
be introduced, and are the topic of this section.

\subsubsection{Getting Rid of Subsumption}
\label{sec:elim-subs}

The type inference algorithm takes in input a \PCF\ term~\t, and returns a typing judgement~\typjudg\ for~\t, 
together with a set \scond\ of so-called \emph{side conditions}. We will show below that \typjudg\ is derivable 
\textit{iff}\/ all the side conditions in \scond\ are valid. Moreover, in this case \typjudg\ is \emph{precise} 
(see Definition~\ref{def:precise}): all occurrences of the base type \Nat{\I}{\J} are in fact of the form 
\NatU{\I}, and the weight and all potentials~\H\ occurring in a sub-type \mtyp{\H}{\A} are kept as low as possible.
Concretely, this means that there is a derivation for~\typjudg\ in which the subsumption rule is restricted to
the following form:
\begin{displaymath}
  \frac{\judg{\fiv}{\ictx}{`G}{\I}{t}{`s}
    \qquad
    \begin{array}[b]{r@{~}l@{\,}l}
      \sjudg{\fiv}{\ictx&}{`D&\teq`G}
      \\
      \sjudg{\fiv}{\ictx&}{`s&\teq`t}
      \\
      \ijudg{\fiv}{\ictx&}{\I&=\J}
    \end{array}
  }{\judg{\fiv}{\ictx}{`D}{\J}{t}{`t}}
\end{displaymath}
The three premises on the right boil down to a set of semantic judgements of the form $\big\{\ijudg{\fiv}{\ictx}{\K_i=\H_i}\big\}$ 
(see Figure~\ref{fig:tsub}), where the $\K_i$'s are indexes occurring in~$`s$ or~$`D$ (or \I\ itself) and the 
$\H_i$'s occur in~$`t$ or~$`G$ (or are~\J\ itself). If the equalities $\K_i=\H_i$ can all be derived
from~\ep, then the three premises  on the right are equivalent to the conjunction (on $i$) of the following properties:
\begin{center}
  ``\itp{\H_i} is defined for any $`r:\fiv→\int$ satisfying \ictx''
\end{center}
(see Section~\ref{sec:sem-judg}).
Given~\ep, this property (called a \emph{side condition}) is denoted by \cjudg{\fiv}{\ictx}{\H_i}.
Actually the type inference algorithm does \emph{not} verify any semantic or subtyping judgement coming
from (instances of) the subsumption rule. Instead, it turns all index equivalences $\H_i=\K_i$ into rewriting rules in~\ep, 
and put all side conditions \cjudg{\fiv}{\ictx}{\H_i} in~\scond. If every side condition in~\scond\ is true 
for~\ep, we write \scval{\scond}. Informally, this means that all subsumptions assumed by 
the algorithm are indeed valid.

\subsubsection{Function Symbols}
\label{sec:symb-var}

Types and judgements manipulated by our type inference algorithm have a very peculiar shape.
In particular, not every index term is allowed to appear in types, and this property will 
be crucial when showing soundness and completeness of the algorithm itself:
\begin{definition}[Primitive Types]
  \label{def:primitive}
  A type is \emph{primitive for~\fiv} when it is on the form
  $\NatU{\f(\fiv)}$, or $\A\toLL\B$ with~\A\ and~\B\ primitive for~\fiv, or \mtyp{\f(\fiv)}{\A} with $a`;\fiv$ 
  and $\A$ primitive for $a;\fiv$. A type is said to be \emph{primitive} when it is primitive for some~\fiv.
\end{definition}
As an example, a primitive type for $\fiv=a;b$ is $\NatU{\f(a,b,c)}\marr{c<\g(a,b)}\NatU{\h(a,b,c)}$. Informally, then, a type is
primitive when the only allowed index terms are function symbols (with the appropriate arity).

\subsubsection{Equational Programs}
\label{sec:ep}

The equational program our algorithm constructs is in fact a \textit{rewriting program}:
every equality corresponds to the (partial) definition of a function symbol, and we may write it 
$\f(a_1,\dots,a_k)`=\J$ (where all free variables of~\J\ are in $\{a_1,\dots,a_k\}$).
If there is no such equation in the rewriting program, we say that \f\ is \emph{unspecified}.

An equational program~\ep\ is \emph{completely specified} if it allows to deduce a precise semantic 
(namely a partial recursive function) for each symbol of its underlying signature 
(written~\sig[\ep]), \ie\ none of the symbols in \sig[\ep] are unspecified.
In other words: a completely specified equational programs has only one \emph{model}.
On the other hand, a \emph{partially specified}  equational program (i.e. a program where symbols 
can possibly be unspecified) can have many models, because partial recursive functions can be assigned 
to function symbols in many different ways, all of them consistent with its equations. Up to now, we 
only worked with completely specified programs, but allowing the possibility to have unspecified symbols 
is crucial for being able to describe the type inference algorithm in a simple way. In the following, \ep\ 
and \eptwo\ denote completely specified equational programs, while \rep\ and \reptwo\ denote rewriting 
programs that are only partially specified.
%
\begin{definition}[Model of a Rewriting Program]
  \label{def:model}
  An \emph{interpretation}~\mapping\ of \rep\ in \ep\ is simply a map from unspecified symbols of \rep\ to 
  indexes on the signature~\sig[\ep], such that if $\f$ has arity $n$, then $\mapping(\f)$ is a term in $\sig[\ep]$ 
  with free variables from $\{a_1,\ldots,a_n\}$. When such an interpretation is defined, we say that \ep\ is 
  a \emph{model} of~\rep, and we write $\mapping:\ep\models\rep$.
\end{definition}
Notice that such an interpretation can naturally be extended to arbitrary index terms on the signature~\sig[\rep], and we assume in 
the following that a rewriting program and its model have disjoint signatures.
\begin{definition}[Validity in a Model]
  \label{def:valid-model}
  Given $\mapping:\ep\models\rep$, we say that a semantic judgement
  \ijudg[\rep]{\fiv}{\ictx}{\I\leq\J} is \emph{valid in the model}
  (notation: \ijudg[\mapping]{\fiv}{\ictx}{\I\leq\J}) when
  \ijudg[\eptwo]{\fiv}{\ictx}{\I\leq\J} where
  $\eptwo=\rep`U\ep`U\{\f`=\mapping(\f)\;|\;
  \f\text{ is unspecified in \rep}\}$.
  This definition is naturally extended to side conditions 
  (with $\scval[\mapping]{\scond}$ standing for $\scval[\eptwo]{\scond}$). 
\end{definition}
Please note that if \rep\ is a completely specified rewriting program, then any model 
$\mapping:\ep\models\rep$ has an interpretation~\mapping\ with an empty domain, and  
$\scval[\mapping]{\scond}$ iff $\scval[\rep]{\scond}$ (still assuming that \sig[\rep] and \sig[\ep] are disjoint).
\lv{
  Now, suppose that $\mapping:\ep\models\rep$ and that the index terms   $\I$ (from $\sig[\ep]$) 
  and $\J$ (from $\sig[\rep]$) are such that $\ijudg{\fiv}{\ictx}{\I\leq\mapping(\J)}$.
  Then, with a slight abuse of notation, we simply write $\ijudg[\mapping]{\fiv}{\ictx}{\I\leq\J}$. 
  The same notation can be extended to types and judgements.
}

As already mentioned, the equational programs handled by our type inference algorithm are not necessarily completely specified.
Function symbols which are not specified are precisely those occurring in ``negative position'' in the judgement produced
in output. This invariant will be very useful and is captured by the following definition:
\begin{definition}[Positive and Negative Symbols]
  Given a primitive type~$`t$, the sets of its \emph{positive} and \emph{negative} symbols (denoted by~\typsymb[+]{`t} and~\typsymb[-]{`t} respectively) are defined inductively by
  \lvsv{
    \begin{displaymath}
      \begin{array}{c@{~=\quad}l@{\qquad}c@{~=\quad}l}
        \typsymb[+]{\NatU{\i(\fiv)}} & \{\i\}; & 
        \typsymb[-]{\NatU{\i(\fiv)}} & ∅; \\
        \typsymb[+]{\arrv{\h(\fiv)}{`s}{`t}}; &
        \typsymb[-]{`s}`U\typsymb[+]{`t}; & 
        \typsymb[-]{\arrv{\h(\fiv)}{`s}{`t}}; &
        \{\h\}`U\typsymb[+]{`s}`U\typsymb[-]{`t}. \\
      \end{array}
    \end{displaymath}
  }{
    \begin{align*}
        \typsymb[+]{\NatU{\i(\fiv)}} &= \{\i\}; \\
        \typsymb[-]{\NatU{\i(\fiv)}} &= ∅; \\
        \typsymb[+]{\arrv{\h(\fiv)}{`s}{`t}} &=
        \typsymb[-]{`s}`U\typsymb[+]{`t}; \\ 
        \typsymb[-]{\arrv{\h(\fiv)}{`s}{`t}}; &=
        \{\h\}`U\typsymb[+]{`s}`U\typsymb[-]{`t}.
    \end{align*}
  }
  Then the set of \emph{positive} (\resp negative) symbols of a judgement~\judg{\fiv}{\ictx}{(x_i:`s_i)_{i\leq n}}{\I}{\t}{`t}
  is the union of all negative (\resp positive) symbols of the~$`s_i$'s and all positive (\resp negative) symbols of~$`t$.
\end{definition}
Polarities in $\{+,-\}$ are indicated with symbols like $\polone,\poltwo$.
Given such a $\polone$, the opposite polarity is $\neg\polone$.
\begin{definition}[Specified Symbols, Types and Judgments]
  \label{def:specif}
  Given a set of function symbols~\symb, a symbol~\f\ is said to be   \emph{\ispec{\symb}{\rep}} 
  when there is a rule~$\f(\fiv)`=\J$ in~\rep\ such that any function symbol appearing in~\J\ is either~\f\ 
  itself, or in~\symb, or a symbol that is \ispec{\symb`U\{\f\}}{\rep\setminus\{\f(\fiv)`=\J\}}.
  Remember that when there is no rule~$\f(\fiv)`=\J$ in~\rep\ the symbol~\f\ is \emph{unspecified} in~$\rep$.
  A primitive type $`s$ is said to be \emph{\tspec{\polone}{\symb}{\rep}} when all function symbols in 
  \typsymb{`s} are \ispec{\symb}{\rep} and all symbols in \typsymb[\neg\polone]{`s} are unspecified.
  A judgement \judg[\rep]{\fiv}{\ictx}{`G}{\I}{t}{`t} is \emph{correctly specified} when~$`t$ and all types 
  in~$`G$ are primitive for~\fiv, and~$`t$ is \tspec{+}{\nsymb}{\rep}, and all types in~$`G$ are 
  \tspec{-}{\nsymb}{\rep}, and all function symbols in~\I\ are \ispec{\nsymb}{\rep} where~\nsymb\ 
  is the set of negative symbols of the judgement.
\end{definition}
In other words, a judgement is correctly specified if the underlying equational program (possibly recursively) defines all 
symbols in positive position depending on those in negative position.

\subsection{The Structure of the Algorithm}

The type inference algorithm that we develop here receives in input a \PCF\ term~\t\ and returns a \dlpcf\ 
judgement~\ejudg{∅}{\K}{\t}{`t} for it, together with a set of side conditions~\scond. We will prove that 
it is \textit{correct}, in the sense that the typing judgement is derivable \emph{iff} the side conditions hold.
The algorithm proceeds as follows:

\begin{varenumerate}
\item\label{algo:fs}
  Compute~\derpcf, a \PCF\ type derivation for~\t;
\item\label{algo:ss} 
  Proceeding by structural induction on~\derpcf, construct a \dlpcf\ derivation for~\t\ 
  (call it \der) and the corresponding set of side conditions~\scond;
\item 
  Returns~\scond\ and the conclusion of~\der.
\end{varenumerate}
%
The \emph{skeleton}~\forget{`s} (or~\forget{\A}) of a modal type~$`s$ (\resp\ of a linear type~$A$) 
is obtained by erasing all its indexes (and its bounds $[a<\I]$). The \emph{skeleton} of a \dlpcf\ derivation 
is obtained by replacing each type by its skeleton, and erasing all the subsumption rules.
In \PCF\ the type inference problem is decidable, and the Step~\ref{algo:fs}. raises no difficulty: actually, one
could even assume that the type \derpcf\ attributes to \t\ is principal.  The core of the algorithm
is of course Step~\ref{algo:ss}. In Section.~\ref{sec:main-fct} we will define a recursive 
algorithm~\mainfct\ that build~\der\ and~\scond\ by annotating \derpcf.
The algorithm~\mainfct\ itself relies on some auxiliary algorithms, which will be described in Section~\ref{sec:aux-alg}
below.

All auxiliary algorithms we will talk about have the ability to generate fresh variables and function symbols.
Strictly speaking, then, they should take a counter (or anything similar) as a parameter, but we elide this 
for the sake of simplicity. Also we consider that we assume the existence of a function \anot{T} that, given a set of index 
variables~$\fiv$ and a \PCF\ type~$T$, returns a modal type~$`t$ primitive for~$\fiv$, containing only fresh 
function symbols, and such that $\forget{`t}=T$.

\subsection{Auxiliary Algorithms and Linear Logic}
\label{sec:aux-alg}


The design of systems of linear dependent types such as \dlpcfv\ and \dlpcfn\ is strongly 
inspired by \BLL, itself a restriction of linear logic. Actually, the best way to present 
the type inference algorithm consists in first of all introducing four auxiliary algorithms, each 
corresponding to a principle regulating the behaviour of the exponential connectives in linear logic. 
Notice that these auxiliary algorithms are the main ingredients of both \dlpcfv\ and \dlpcfn\ type inference.
Consistently to what we have done so far, we will prove and explain them with \dlpcfv\ in mind. All the auxiliary algorithm
we will talk about in this section will take a tuple of \dlpcfv\ types as first argument; we assume that all of them have the
same skeleton and, moreover, that all index terms appearing in them are pairwise distinct.



\paragraph{Dereliction.}
Dereliction is the following principle: any duplicable object (say, of type $!\A$) 
can be made linear (of type $\A$), that is to say $!\A\rightarrow \A$. 
In \dlpcf, being duplicable means having a modal type, which also contains some quantitative information, namely 
how many times the object can be duplicated, at most. In \dlpcf, dereliction can be simply seen as the principle 
$\mtyp{1}{\A}\rightarrow \A\isubst{0}$, and is implicitly used in the rules (\textit{App}) and (\textit{Fix}).
Along the type inference process, as a consequence, we often need to create ``fresh instances'' of dereliction
in the form of pairs of types being in the correct semantic relation. This is indeed possible:
\begin{lemma}
  \label{lem:der}
  There is an algorithm $\algder$ such that given two types~$`t$
  (primitive for~$\fiv$) and~$`s$ (primitive for~$a,\fiv$) of the same   skeleton, 
  $\algderp{`s}{`t}{a}{\fiv}{\ictx}{\polone}=(\rep,\scond)$ where:
  \begin{varenumerate}
  \item
    for every~$\ep\supseteq\rep$, if
    $\scval{\scond}$\;then\;
    $\sjudg{\fiv}{\ictx}{`s\isubst{\zero}\teq`t}$;
  \item
    whenever $\sjudg{\fiv}{\ictx}{`u\isubst{\zero}\teq`x}$ where
    $\forget{`u}\equiv\forget{`s}$, there is    
    $\mapping:\ep\models\rep$ such that 
    $\sjudg[\mapping]{\fiv}{\ictx}{`s\teq`u}$,
    $\sjudg[\mapping]{\fiv}{\ictx}{`t\teq`x}$,   
    and $\scval[\mapping]{\scond}$;
  \item
    $`s$ is $\tspec{\polone}{\typsymb[\polone]{`t}}{\rep}$ and
    $`t$ is $\tspec{\neg\polone}{\typsymb[\neg\polone]{`s}}{\rep}$.
  \end{varenumerate}
\end{lemma}
\begin{proof}
The algorithm $\algder$ is defined by recursion on the structure of $`s$:
\begin{varitemize}
\item
  Here are two base cases:
  \begin{align*}
  \algderp{\NatU{\f(a,\fiv)}}{\NatU{\g(\fiv)}}{a}{\fiv}{\ictx}{+}&=(\{\g(\fiv)`=\f(0,\fiv)\},\{\fiv,\ictx\models\f(0,\fiv)\downarrow\})\\
  \algderp{\NatU{\f(a,\fiv)}}{\NatU{\g(\fiv)}}{a}{\fiv}{\ictx}{-}&=(\{\f(a,\fiv)`=\g(\fiv)\},\{\fiv,\ictx\models\g(\fiv)\downarrow\})
  \end{align*}
\item
  Inductive cases are slightly more difficult:
  \begin{align*}
  \algderp{`s_1\marr{b<\f(a,\fiv)}`s_2}{`t_1\marr{b<\g(\fiv)}`t_2}{a}{\fiv}{\ictx}{+}&=(\rep^-_1\cup\rep^+_2\cup\{\g(\fiv)`=\f(0,\fiv)\},\\
       &\qquad\scond^-_1\cup\scond_2^+\cup\{\fiv,\ictx\models\f(0,\fiv)\downarrow\});\\
  \algderp{`s_1\marr{b<\f(a,\fiv)}`s_2}{`t_1\marr{b<\g(\fiv)}`t_2}{a}{\fiv}{\ictx}{-}&=(\rep^+_1\cup\rep^-_2\cup\{\f(a,\fiv)`=\g(\fiv)\},\\
       &\qquad\scond^-_1\cup\scond_2^+\cup\{\fiv,\ictx\models\g(\fiv)\downarrow\}).    
  \end{align*}
  where:
  \begin{align*}
  \algderp{`s_1}{`t_1}{a}{b,\fiv}{\ictx,b<\f(0,\fiv)}{+}=(\rep^-_1,\scond^-_1);\\
  \algderp{`s_1}{`t_1}{a}{b,\fiv}{\ictx,b<\g(\fiv)}{-}=(\rep^+_1,\scond^+_1);\\
  \algderp{`s_2}{`t_2}{a}{b,\fiv}{\ictx,b<\f(0,\fiv)}{+}=(\rep^+_2,\scond^+_2);\\
  \algderp{`s_2}{`t_2}{a}{b,\fiv}{\ictx,b<\g(\fiv)}{-}=(\rep^-_2,\scond^-_2).
  \end{align*}
\end{varitemize}
Let us now prove the Lemma by induction on the structure of $`s$:
\begin{varitemize}
\item
  If $`s=\NatU{\f(0,\fiv}$, then the thesis can be easily reached;
\item
  Suppose that $`s=`s_1\marr{b<\f(a,\fiv)}`s_2$ and that $`t=`t_1\marr{b<\g(\fiv)}`t_2$. For simplicity, suppose that $\polone=+$ (the
  case $\polone=-$ is analogous):
  \begin{varenumerate}
  \item
    If $\ep\supseteq\rep$, then $\ep\supseteq\rep^-_1$ and $\ep\supseteq\rep^+_2$. Moreover,
    if $\scval{\scond}$ then, of course $\scval{\scond_1^-}$ and $\scval{\scond_2^+}$. As a consequence,
    by induction hypothesis, we have
    \begin{align*}
      \sjudg{\fiv,b}{\ictx,b<\g(\fiv)}{&`s_1\isubst{\zero}\teq`t_1};\\
      \sjudg{\fiv,b}{\ictx,b<\f(a,\fiv)}{&`s_2\isubst{\zero}\teq`t_2}.
    \end{align*}
    From $\fiv;\ictx\models_\ep\f(0,\fiv)\downarrow$ it follows that $\fiv;\ictx\models_\ep\g(\fiv)=\f(0,\fiv)$. As a consequence,
    $$
    \sjudg{\fiv}{\ictx}{`s\teq`t}.
    $$
  \item
    If $\sjudg{\fiv}{\ictx}{`u\isubst{\zero}\teq`x}$ where
    $\forget{`u}\equiv\forget{`s}$, we can safely assume that
    \begin{align*}
      `u&=`u_1\marr{b<\I}`u_2\\
      `x&=`x_1\marr{b<\J}`x_2
    \end{align*}
    and moreover, that
    \begin{align*}
      \sjudg{\fiv,b}{\ictx,b<\I\isubst{\zero}}{&`u_1\isubst{\zero}\teq`x_1\isubst{\zero}};\\
      \sjudg{\fiv,b}{\ictx,b<\J}{&`u_2\teq`x_2};\\
      \ijudg{\fiv}{\ictx}{&\I\isubst{\zero}=\J}.
    \end{align*}
    Assume for the sake of simplicity that $\polone=+$. Then, from the two judgements above, one
    obtains:
    \begin{align*}
      \sjudg[\ep\cup\{\f(a,\fiv)`=\I\}]{\fiv,b}{\ictx,b<\f(\zero,\fiv)&}{`u_1\isubst{\zero}\teq`x_1\isubst{\zero}}\\
      \sjudg[\ep\cup\{\g(\fiv)`=\J\}]{\fiv,b}{\ictx,b<\g(\fiv)&}{`u_2\teq`x_2}
    \end{align*}    
    By induction hypothesis, there are $\mapping_1:(\ep\cup\{\f(a,\fiv)`=\I\})\models\rep_1^-$ and $\mapping_2:(\ep\cup\{\g(\fiv)`=\J\})\models
    \rep_2^+$ such that
    \begin{align*}
      \sjudg[\ep\cup\{\f(a,\fiv)`=\I\}]{\fiv,b}{\ictx,b<\f(\zero,\fiv)&}{`s_1\isubst{\zero}\teq`u_1\isubst{\zero}};\\
      \sjudg[\ep\cup\{\f(a,\fiv)`=\I\}]{\fiv,b}{\ictx,b<\f(\zero,\fiv)&}{`t_1\isubst{\zero}\teq`x_1\isubst{\zero}};\\
      \sjudg[\ep\cup\{\g(\fiv)`=\J\}]{\fiv,b}{\ictx,b<\g(\fiv)&}{`s_2\teq`u_2};\\
      \sjudg[\ep\cup\{\g(\fiv)`=\J\}]{\fiv,b}{\ictx,b<\g(\fiv)&}{`t_2\teq`x_2}.
    \end{align*}
    and $\scval[\mapping_1]{\scond_1^-}$, and $\scval[\mapping_2]{\scond_2^+}$. By way of the push-out
    technique, we get $\mappingtwo_1:\ep\models\rep_1^-$ and $\mappingtwo_2:\ep\models\rep_2^+$ such that
    $\scval[\mappingtwo_1]{\scond_1^-}$, and $\scval[\mappingtwo_2]{\scond_2^+}$.
    Now, $\mapping$ is defined as
    $\mapping_1\cup\mapping_2$, plus the assignment of $\I$ (after some $\alpha$-renaming) to $\f$.
    First of all, notice that $\scval[\mapping]{\scond}$, simply because $\ijudg{\fiv}{\ictx}{\I\isubst{\zero}\downarrow}$.
    Moreover, observe that, by easy manipulations of the derivations above, one gets:
    \begin{align*}
      \sjudg[\mapping]{\fiv,b}{\ictx,b<\I\isubst{\zero}&}{`s_1\isubst{\zero}\teq`u_1\isubst{\zero}};\\
      \sjudg[\mapping]{\fiv,b}{\ictx,b<\I\isubst{\zero}&}{`t_1\isubst{\zero}\teq`x_1\isubst{\zero}};\\
      \sjudg[\mapping]{\fiv,b}{\ictx,b<\J&}{`s_2\teq`u_2};\\
      \sjudg[\mapping]{\fiv,b}{\ictx,b<\J&}{`t_2\teq`x_2}.
    \end{align*}
  \item
    This is an easy consequence of how $`s$ and $`t$ are built.
  \end{varenumerate}
\end{varitemize}
This concludes the proof.
\end{proof}
The algorithm $\algder$ works by recursion on the \PCF\ type $\forget{`s}$ and has thus linear complexity
in $|\forget{`s}|$.

\paragraph{Contraction.}
Another key principle in linear logic is contraction, according
to which two copies of a duplicable object can actually be produced, 
$!\A\rightarrow !\A\otimes !\A$. Contraction is used in binary rules 
like (\textit{App}) or (\textit{If}), in the form of the operator~$\uplus$. 
This time, we need an algorithm $\algcontr$ which takes \emph{three} linear types $\A$, $\B$
and $\C$ (all of them primitive for $(a,\fiv)$) and turn them into an equational program and a set of 
side conditions:
$$
\algcontrp{\A}{\B}{\C}{\I}{\J}{a}{\fiv}{\ictx}{\polone}=(\rep,\scond).
$$
The parameters $\I$ and $\J$ are index terms capturing the number of times $\B$ and $\C$
can be copied. A Lemma akin to~\ref{lem:der} can indeed be proved about $\algcontr$.
\sv{
In particular, for any $\ep\supseteq\rep$, if $\scval{\scond}$ then
\begin{equation}
  \label{eq:contr}
  \sjudg{\fiv}{\ictx}{\mtyp{\I+\J}{\A}\teq(\mtyp{\I}{\B'})\uplus(\mtyp{\J}{\C'})}
\end{equation}
for some $\B'$ and $\C'$ such that $\sjudg{\fiv}{\ictx,a<\I}{\B'\teq\B}$
and $\sjudg{\fiv}{\ictx,a<\J}{\C'\teq\C}$.
}
\lv{
\begin{lemma}
  \label{lem:contr}
  There is an algorithm $\algcontr$ such that whenever \B\ and~\C\ are two linear 
  types of the same skeleton and primitive for~$(a,\fiv)$, $\algcontrp{\B}{\C}{\I}{\J}{a}{\fiv}{\ictx}{\polone}=
  (\A,\rep,\scond)$, where:
  \begin{varitemize}
  \item $\A$ is primitive for $a,\fiv$, and
    $\forget{\A}=\forget{\B}=\forget{\C}$;
  \item
    For any $\ep\supseteq\rep$,\quad $\scval{\scond}$\;iff\;
    $$\sjudg{\fiv}{\ictx}{\mtyp{\I+\J}{\A}\teq(\mtyp{\I}{\B'})\uplus(\mtyp{\J}{\C'})}$$
    for some $\B'$ and $\C'$ such that
    $\sjudg{\fiv}{\ictx,a<\I}{\B'\teq\B}$
    and $\sjudg{\fiv}{\ictx,a<\J}{\C'\teq\C}$.
  \item
    \A\ is
    \tspec{\polone}{\typsymb[\polone]{\B}`U\typsymb[\polone]{\C}}{\rep},
    and \B\ and \C\ are
    \tspec{\neg\polone}{\typsymb[\neg\polone]{\A}}{\rep}.
\end{varitemize}
\end{lemma}
}

\paragraph{Digging.}
In linear logic, any duplicable object having type $!A$ can be turned into an object of type $!!A$, namely an object which is the duplicable version of a duplicable object.
Digging is the principle according to which this transformation is possible, namely $!\A\rightarrow !!\A$.
At the quantitative level, this corresponds to splitting a bounded sum into its summands.
This is used in the typing rules for functions,  ($\multimap$) and (\textit{Fix}).

The auxiliary algorithm corresponding to the digging principle takes two linear types and builds, as usual, a rewriting program and a set of side conditions capturing the fact that the first of the two types is the bounded sum of the second:
$$
\algdigp{\A}{\B}{\I}{\J}{\fiv}{a}{b}{\ictx}{\polone}~=~(\rep,\scond).
$$
The correctness of $\algdig$ can again be proved similarly to what we did in Lemma~\ref{lem:der},
the key statement being that for every $\ep\supseteq\rep$ such that $\scval{\scond}$,
the following must hold
$$
\sjudg{\fiv}{\ictx}{\mtyp[b]{\sum_{a<\I}\J}{\A}\teq
  \sum_{a<I}\mtyp[b]{J}{\C}}
$$
for some $\C$ such that
$\sjudg{\fiv}{\ictx,a<\I,b<\J}{\C\teq\B}$.
\lv{
\begin{lemma}
  \label{lem:digging}
  There is an algorithm $\algdig$ such that for any linear type~\B\  primitive for~$(a,b,\fiv)$,
  \lv{\\}
  $\algdigp{\B}{\I}{\J}{\fiv}{a}{b}{\ictx}{\polone}~=~(\A,\rep,\scond)$,
  where:
  \begin{varitemize}
  \item $\A$ is primitive for $b,\fiv$, and
    $\forget{\A}=\forget{\B}$;
  \item
    for any $\ep\supseteq\rep$,\quad $\scval{\scond}$ iff
    $$
    \sjudg{\fiv}{\ictx}{\mtyp[b]{\sum_{a<\I}\J}{\A}\teq
      \sum_{a<I}\mtyp[b]{J}{\B'}}
    $$
    for some $\B'$ such that
    $\sjudg{\fiv}{\ictx,a<\I,b<\J}{\B'\teq\B}$;
  \item
    \A\ is \tspec{\polone}{\typsymb[\polone]{\B}}{\rep} and 
    \B\ is \tspec{\neg\polone}{\typsymb[\neg\polone]{\A}}{\rep}.
  \end{varitemize}
\end{lemma}
}

\paragraph{Weakening.}
Weakening means that duplicable objects can also be erased, even when the underlying index is~\zero.
Weakening is useful in the rules (\textit{Ax}) and (\textit{n}).
Once a fresh \dlpcfv\ type $\A$ is produced, the only thing we need to do is to produce an equational 
program $\rep$ specifying (in an arbitrary way) the symbols in $\A^\polone$, this way preserving the 
crucial invariants about the equational programs manipulated by the algorithm. Formally, it means 
that there is an algorithm $\algweak$ such that
$$
\algweakp{\A}{\fiv}{a}{\polone}=\rep,
$$
where \A\ is \tspec{\polone}{∅}{\rep}
\lv{
\begin{lemma}
  \label{lem:weak}
  There is  an algorithm $\algweak$ such that   $\algweakp{\tpcfone}{\fiv}{a}{\polone}=(\A,\rep)$ where,
\begin{varitemize}
\item \A\ is primitive for $a;\fiv$, and $\forget{\A}=\tpcfone$;
\item
  for every $\ep\supseteq\rep$, it holds that
  $\sjudg{\fiv}{\ictx,a<0}{\A\teq\A}$;
\item
  \A\ is \tspec{\polone}{∅}{\rep}.
\end{varitemize}
\end{lemma}}
Observe how no sets of constraints is produced in output by $\algweak$, contrarily to $\algder$ and $\algcontr$.

\lv{
\paragraph{Some Extra Functions.}
There are also some extra auxiliary algorithms that are used by the function~\mainfct.
Although they are not related to some exponential rule of~\LL, they are defined in a similar way than the previous one.

\noindent
The following lemma is used for type checking the (\textit{Ax}) rule:
\begin{lemma}
  \label{lem:ax}
  There is  an algorithm $\algax$ such that, given two types~$`s$ and~$`t$ primitive for~\fiv\ and of same skeleton,
  $\algaxp{`t}{`s}{\fiv}{\ictx}{\polone}=(\rep;\scond)$ where,
  \begin{varitemize}
  \item
    for every $\ep\supseteq\rep$,\quad
    $\scval{\scond}\iff\sjudg{\fiv}{\ictx}{`t\teq`s}$;
  \item $`t$ is \tspec{\polone}{\typsymb[\neg\polone]{`s}}{\rep} and
    $`s$ is \tspec{\neg\polone}{\typsymb[\polone]{`t}}{\rep}.
  \end{varitemize}
\end{lemma}

We also need a generalisation of the dereliction algorithm, where any index~\I\ can be used instead of~\zero\ in~$\algder$.
It is necessary for type checking a fix point.
\begin{lemma}
  \label{lem:sbs}
  There is  an algorithm $\algsub$ such that, for any type~$`t$ primitive for~$\fiv$ and~$`s$ primitive for~$(\fiv,a)$ that have the same skeleton,
  $\algsubp{`s}{`t}{a}{\I}{\fiv}{\ictx}{\polone}=(\rep,\scond)$ where,
  \begin{varitemize}
  \item
    for any~$\ep\supseteq\rep$,\quad
    $\scval{\scond}\iff\sjudg{\fiv}{\ictx}{`s\isubst{\I}\teq`t}$;
  \item
    $`s$ is \tspec{\polone}{\typsymb[\polone]{`t}}{\rep} and
    $`t$ is \tspec{\neg\polone}{\typsymb[\neg\polone]{`s}}{\rep}.
  \end{varitemize}
\end{lemma}

The following lemma is needed to type check the (\textit{If}) rule.
It requires to use a function symbol $\mathrm{if}(a,b,c)$ with two equations:\quad
$\mathrm{if}(\zero,b,c)=b$\ \ and\ \ $\mathrm{if}(a+\one,b,c)=c$.
\begin{lemma}
  \label{lem:if}
  There is  an algorithm $\algif$ such that for any types~$`t_1,`t_2$ with the same skeleton, that are primitive for~$\fiv$,\\
  $\algifp{`t_1}{`t_2}{\I}{\fiv}{\ictx}{\polone}=(`s,\rep,\scond)$ where,
  \begin{varitemize}
  \item $`s$ is primitive for $(a,\fiv)$, and
    $\forget{`s}=\forget{`t_i}$;
  \item
    For any equational program $\ep\supseteq\rep$, \quad
    $\scval{\scond}$ \textit{iff} \\
    $\sjudg{\fiv}{\ictx,\I=\zero}{`t_1\teq`s}$ and
    $\sjudg{\fiv}{\ictx,\I\geq\one}{`t_2\teq`s}$;
  \item
    $`s$ is \tspec{\polone}{\typsymb[\polone]{`t_1}`U\typsymb[\polone]{`t_2}}{\rep},
    and $`t_1$ and $`t_2$ are \tspec{\neg\polone}{\typsymb[\neg\polone]{`s}}{\rep}.
  \end{varitemize}
\end{lemma}
}


\subsection{The Type Inference Procedure}
\label{sec:main-fct}

In this section, we will describe the core of our type inference
algorithm.
This consists in a recursive algorithm \mainfct\ which decorates a \PCF\ type derivation \derpcf, producing in output a
\dlpcf\ judgement, together with an equational program and a set of
side conditions. In order to correctly create fresh symbols and to format 
side conditions properly, the main recursive function~\mainfct\ also receives 
a set of index variables~$\fiv$ and a set of constraints~$\ictx$ in input.
Thus, it has the following signature:
\begin{displaymath}
  \mainfct(\fiv,\ictx,\derpcf)\;=\;(\ljudg{`G}{\I}{t}{`t};\rep;\scond).
\end{displaymath}
We will prove that the the output of $\mainfct$ satisfies the following two invariants:
\begin{varitemize}
\item
  \emph{Decoration}. $\forget{`G}\vdash\t:\forget{`t}$ is the judgement concluding~\derpcf.
\item
  \emph{Polarity}. \judg[\rep]{\fiv}{\ictx}{`G}{\I}{\t}{`t} is correctly specified (Definition~\ref{def:specif}).
\end{varitemize}
The algorithm~\mainfct\ proceeds by inspecting~\derpcf\ in an inductive manner.
It first annotates the types in the conclusion  judgement with fresh function symbols to get a \dlpcf\ judgement~\typjudg.
Then a recursive call is performed on the  immediate sub-derivations of $\derpcf$, this way obtaining 
some \dlpcf\ typing judgement~$\typjudg_i$. Finally \mainfct\ generates, calling the auxiliary algorithms, 
the equations on function symbols that allow to derive~\typjudg\ from the~$\typjudg_i$'s, 
The equations are written in~\rep, and the required assumptions of index convergence in~\scond.

Decoration and Polarity are the invariants of the algorithm~\mainfct.
In particular, the auxiliary algorithms are always called with the appropriate parameters,
this way enforcing Polarity. 

The algorithm computing~\mainfct\ proceeds by case analysis on~\derpcf.
\sv{We detail some cases here, the other ones are developed in~\cite{LV}.}
\begin{varitemize}
\lv{
\item 
  Assume
  \begin{math}\displaystyle
    \derpcf=\frac{}{y_1:U_1,\dots,y_k:U_k,x:T\vdash x:T}\lv{(\textit{Ax})}~.
  \end{math}
  \\
  For each~$i$, let $\B_i=\anot{U_i}$ and
  $\rep_i=\algweakp{\B_i}{\fiv}{b_i}{-}$
  (where all the $b_i$'s   are fresh).
  Then let~$`s$ and~$`t$ produced by \anot{T}, and write
  $(\rep_0;\scond)=\algaxp{`t}{`s}{\fiv}{\ictx}{+}$.
  Thus return
  $(\ljudg{`G,x:`s}{\zero}{x}{`t};\rep;\scond)$ with

  \hfil
  \begin{math}
    \left\{
    \begin{array}{c@{~=~}l}
      \rep & \rep_0`U\bigcup_{i\leq k}(\rep_i`U\{\h_i`=\zero\}) \\
      `G & \{y_i:\mtyp[b_i]{\h_i}{\B_i}\}_{i\leq k}\\
    \end{array}
    \right.
  \end{math}
  \\ where the $\h_i$'s are fresh.
}
\item 
  Suppose 
  \begin{math}\displaystyle
    \derpcf=\frac{}{y_1:U_1,\dots,y_k:U_k\vdash \nb:\NatPCF}\lv{(n)}~.
  \end{math}
  \\
  \lvsv{Again,}{For each~$i$,} let $\B_i=\anot{U_i}$ and
  $\rep_i=\algweakp{\B_i}{\fiv}{b_i}{-}$ 
  \sv{(where all the $b_i$'s are fresh)}. Let $\i(\fiv)$ be a fresh function symbol.
  Then return $(\ljudg{`G}{\zero}{\nb}{\NatU{\i(\fiv)}};\rep;∅)$ 
  where the $\h_i$'s are fresh symbols and
  \begin{align*}
      \rep &= \bigcup_i(\rep_i`U\{\h_i(\fiv)`=\zero\})
      `U\{~\i(\fiv)`=\inb~\};\\
      `G &= \{y_i:\mtyp[b_i]{\h_i(\fiv)}{\B_i}\}_{i\leq k}.
    \end{align*}
%
\lv{
\item 
  If 
  \begin{math}\displaystyle
    \derpcf=\frac{(\derpcf'):
      \begin{array}[b]{c}
        \vdots \\
        `P\vdash\t:\NatPCF\quad
      \end{array}
    }{`P\vdash\suc(\t):\NatPCF}(s)~.
  \end{math}
  \\
  Let
  $(\ljudg{`G}{\K}{\t}{\NatU{\j(\fiv)}};\rep_0;\scond)=
  \mainfct(\derpcf',\fiv,\ictx)$, and
  let $\i(\fiv)$ be a fresh symbol.
  Then return 
  $(\ljudg{`G}{\K}{\suc(\t)}{\NatU{\i(\fiv)}};\rep;\scond)$ where
  $\rep = \rep_0`U\{~\i(\fiv)`=\j(\fiv)+\one~\}$.
\item 
  For the typing rule~(\textit{p}), we do the same as previously but with $\i(\fiv)=\j(\fiv)\mnu\one$ instead of $\i(\fiv)=\j(\fiv)+\one$ in the equational program.

\item 
  If \derpcf=
  \begin{math}\displaystyle
    \frac{({\derpcf}_0):
      \begin{array}[b]{c}
        \vdots \\
        `P\vdash\t:\NatPCF
      \end{array}
      \quad
      \begin{array}[b]{lc}
        & \vdots \\
        ({\derpcf}_1): & `P\vdash\u_1:T \\
      \end{array}
      \quad
      \begin{array}[b]{lc}
        & \vdots \\
        ({\derpcf}_2): & `P\vdash\u_2:T
      \end{array}
    }{`P\vdash\ifz{\t}{\u_1}{\u_2}:T}(\textit{If})~,
  \end{math}
   
  let $(\ljudg{`G}{\K}{\t}{\NatU{\j(\fiv)}};\rep_0;\scond_0)=
  \mainfct({\derpcf}_0,\fiv,\ictx)$,
  and $(\ljudg{`D_i}{\H_i}{\u_i}{`t_i};\rep_i;\scond_i)=
  \mainfct({\derpcf}_i,\fiv,\ictx_i)$ (for $i=1,2$),
  with $\ictx_1=(\ictx,\j(\fiv)=\zero)$ and
  $\ictx_2=(\ictx,\j(\fiv)\geq\one)$.
  $`t_1$ and~$`t_2$ both have skeleton~$T$ and are primitive for~$\fiv$ (since Decoration and Polarity hold by induction hypothesis).
  So we can compute\quad
  $(`t';\rep';\scond')=\algifp{`t_1}{`t_2}{\j(\fiv)}{\fiv}{\ictx}{+}$.
  \\
  In the same way, $`D_1$ and $`D_2$ have the same skeleton~$`P$.
  So for each $(y_i:`s_i)$ in~$`D_1$ there is $(y_i:`s_i')$ in~$`D_2$ so that we can compute
$(`s_i'',\rep_i',\scond_i')=\algifp{`s_i}{`s_i'}{\j(\fiv)}{\fiv}{\ictx}{-}$.
  Let $`D=\{y_i:`s_i''\}$ and write $\mtyp[a_i]{\m_i}{\C_i}=`s_i''$.
  $`D$ also has skeleton~$`P$, and so has~$`G$.
  So for each~$i$, there is some variable declaration $y_i:\mtyp[a_i]{\n_i}{\B_i}$ in~$`G$, so that we can compute
$(\A_i,\rep_i'',\scond_i'')=\algcontrp{\B_i}{\C_i}{\n_i}{\m_i}{a_i}{\fiv}{\ictx}{-}$.

  Then let $\h_i(\fiv)$ be a fresh symbol for each~$i$, and return
  $(\ljudg{`G'}{\K+\H}{\ifz{\t}{\u_1}{\u_2}}{`t'};\rep;\scond)$ where
  
  \hfil
  \begin{math}
    \begin{array}[t]{ccl}
      \scond &=& \scond_0 `U \scond_1 `U \scond_2`U\scond'
      `U\bigcup_i(\scond_i'`U\scond_i'')
      \\
      \rep &=& \rep_0 `U \rep_1 `U \rep_2`U\rep'
      `U \{~\h_i(\fiv)`=\n_i+\m_i~\}_i \\
      && `U\bigcup_i(\rep_i'`U\rep_i'')
      \\
      `G' &=& \{y_i:\mtyp[a_i]{\h_i(`f)}{\A_i}\}_i\\
      \H &=& \mathrm{if}(\j(\fiv),\H_1,\H_2) \\
    \end{array}
  \end{math}
}
  
\item 
  If \derpcf\lvsv{=}{ is on the form\\}
  $$
  \infer[]
  {`P\vdash\t~\u:T}
  {
        \derpcf^1:`P\vdash\t:U\toPCF T
    & 
        \derpcf^2: `P\vdash\u:U 
   }
   $$
  let  
  $(\ljudg{`G_1}{\K}{\t}{\arrv{\f(\fiv)}{`s_1}{`s_2}};\rep_1;\scond_1)=
  \mainfct(\fiv;\ictx;{\derpcf}_1)$, 
  and
  $(\ljudg{`G_2}{\H}{\u}{`t} ;\ \rep_2 ;\ \scond_2)=
  \sv{\linebreak[4]}
  \mainfct(\fiv;\ictx;{\derpcf}_2)$.
  Let $(\reptwo,\scondtwo)=\algderp{`s_1}{`t}{a}{\fiv}{\ictx}{+}$.
  We then annotate~$T$: let $`t_2=\anot{T}$, and
  let $(\repthree,\scondthree)=\algderp{`s_2}{`t_2}{a}{\fiv}{\ictx}{-}$.
  Then we build a context equivalent to $`G_1\uplus`G_2$:
  \sv{by the decoration property, $`G_1$ and $`G_2$ have the same skeleton~$`P$, so}
  for any $y:\mtyp[b_y]{\i_y(\fiv)}{\B_y}$ in~$`G_1$, 
  there is some   $y:\mtyp[b_y]{\j_y(\fiv)}{\C_y}$ in~$`G_2$
  (possibly after some~$`a$-conversion).
  Then let $\A_y=\anot[(b_y,\fiv)]{\forget{\B_y}}$, and
  $(\rep_y;\scond_y)=
  \algcontrp{\A_y}{\B_y}{\C_y}{\i_y(\fiv)}{\j_y(\fiv)}{b_y}{\fiv}{\ictx}{-}$.
  There are $`D_i\teq`G_i$ (for $i=1,2$)
  such that
  $\sjudg{\fiv}{\ictx}{\{y:\mtyp[b_y]{\i_y(\fiv)+\j_y(\fiv)}{\A_y}\}_{y}\teq`D_1\uplus`D_2}$,
  for every
  $\ep\supseteq\bigcup_y\rep_y$ such that
  $\scval{(\bigcup_y\scond_y)}$. 
  Thus let $\h_y$'s be fresh symbols and return
  $(\ljudg{`D}{\K+\H}{\t\,\u}{`t_2};\rep;\scond)$ with
%
  \begin{align*}
      \scond =&\; \scond_1 `U \scond_2 `U \scondtwo `U \scondthree \\
      &`U\bigcup_{y}\big(\scond_y\cup\{\cjudg{\fiv}{\ictx}{\i_y(\fiv)+\j_y(\fiv)}\}\big);\\
      \rep =&\; \rep_1`U\rep_2`U\reptwo`U\repthree`U\{\f(\fiv)`=\one\} \\
      &`U\bigcup_y\big(\rep_y\cup\{\h_y(\fiv)`=\i_y(\fiv)+\j_y(\fiv)\}\big);\\
      `D =&\; \{y:\mtyp[b_y]{\h_y(\fiv)}{\A_y}\}_{y}.
 \end{align*}
%
\item 
  Assume that $\derpcf$ is
  $$
  \infer[]
  {`P\vdash`lx.\t:U\toPCF T}
  {\derpcf':`P,x:U\vdash\t:T}
  $$
  Let $a$ be a fresh index variable, and~$\i(\fiv)$ be a fresh function symbol, and compute
  $(\ljudg{`G,x:`s}{\K}{\t}{`t};\reptwo;\scondtwo)=
  \mainfct(\derpcf',(a,\fiv),(a<\i(\fiv),\ictx))$.
  We build a context equivalent to $\sum_{a<\i(\fiv)}`G$:
  for every $y:\mtyp[b_y]{\j_y(a,\fiv)}{\B_y}`:`G$,
  let \sv{\linebreak[4]}$A_y=\anot[(b_y,\fiv)]{\forget{\B_y}}$, 
  let~$\h_y(\fiv)$ be a fresh symbol, 
  and write
  $(\rep_y,\scond_y)=\algdigp{\A_y}{\B_y}{\i(\fiv)}{\j_y(a,\fiv)}{\fiv}{a}{b_y}{\ictx}{-}$.
  Then return 
  $$
  (\ljudg{`D}{\i(\fiv)+\sum_{a<\i(\fiv)}\K}{`lx.\t}{\arrv{\i(\fiv)}{`s}{`t}};\rep;\scond)
  $$  
  where
  %
  \begin{align*}
    \scond &= \scondtwo \cup\bigcup_y(\scond_y\cup\{\ijudg{\fiv}{\ictx}{\sum_{a<\i(\fiv)}\j_y(a,\fiv)`|}\});\\
    \rep &=\reptwo\cup\cup_y(\rep_y\cup\{\h_y(\fiv)`=\sum_{a<\i(\fiv)}\j_y(a,\fiv)\});\\
      `D &= \{y:\mtyp[b_y]{\h_y(\fiv)}{\A_y}\}_y.
  \end{align*}
\lv{
\item 
  If
  \begin{math}\displaystyle
    \derpcf=\frac{(\derpcf'):
      \begin{array}[b]{c}
        \vdots \\
        `P,x:T\vdash\t:T
      \end{array}
    }{`P\vdash\fix{\t}:T}(\textit{Fix})
  \end{math}~.
  \\
  Let $b$ be a fresh index variable, and~$\h(\fiv)$ be a fresh function symbol, and compute
  $(\ljudg{`G,x:\mtyp{\i(b,\fiv)}{\A}}{\J}{\t}{\mtyp{\m(b,\fiv)}{\B}};\rep';\scond')=
  \mainfct(\derpcf',(b,\fiv),(b<\h(\fiv),\ictx))$
  (up to $`a$-conversion we can assume that the index variable~$a$ is   the same in the context and the type).
  Because of the type equations we have to ensure (namely
  \sjudg{(a,b,\fiv)}{(a<\i(b,\fiv),b<\h(\fiv),\ictx)}{\B\isubst{\zero}\isubst[b]{\one+b+\fc{b+1}{a}{\i(b,\fiv)}}\teq\A}), 
  we will use intermediate types.
  \\
  First we annotate~$T$: let $`s=\anot[b,\fiv]{T}$.
  Up to $`a$-conversion, $`s=\mtyp{\l(b,\fiv)}{\C}$ for some~\C.
  Let
  $(\rep_1,\scond_1)=\algderp{\B}{\C}{a}{(b,\fiv)}{b<\h(\fiv),\ictx}{-}$.
  By Lemma~\ref{lem:der}, \C\ is \tspec{+}{\typsymb[+]{\B}}{\rep_1}
  and \B\ is \tspec{-}{\typsymb[-]{\C}}{\rep_1}, and 
  for any~$\ep\supseteq\rep_1$, 
  \begin{displaymath}
    \scval{\scond_1} \quad\iff\quad
    \sjudg{b,\fiv}{b<\h(\fiv),\ictx}{\B\isubst{\zero}\teq\C}~.
  \end{displaymath}
  Now let
  $(\rep_2,\scond_2)=
  \algsubp{\C}{\A}{b}{\one+b+\fc{b+1}{a}{\i(b,\fiv)}}{(a,b,\fiv)}{a<\i(b,\fiv),b<\h(\fiv),\ictx}{-}$.
  By Lemma~\ref{lem:sbs}, \A\ is \tspec{+}{\typsymb[+]{\C}}{\rep_2}
  and \C\ is \tspec{-}{\typsymb[-]{\A}}{\rep_2}, and 
  for any~$\ep\supseteq\rep_2$,\ $\scval{\scond_2}$ iff
  \begin{displaymath}
    \sjudg{a,b,\fiv}{a<\i(b,\fiv),b<\h(\fiv),\ictx}{
      \C\isubst[b]{\one+b+\fc{b+1}{a}{\i(b,\fiv)}}\teq\A}~. 
  \end{displaymath}
  \TODO{Here}
  
  Hence $\ep_1`U\ep_2$ is positively defined for~\A\ and negatively for~\B, and~$\ep_1`U\ep_2``(=\eptwo\vDash(\scond_1,\scond_2)$ implies
  \begin{displaymath}
    \sjudg[\eptwo]{a,b,\fiv}{a<\i(b,\fiv),b<\h(\fiv),\ictx}{\B\isubst{\zero}\isubst[b]{\one+b+\fc{b+1}{a}{\i(b,\fiv)}}\teq\A}.
  \end{displaymath}
  Again we annotate~$T$: let $`t=\anot[\fiv]{T}$.
  Up to $`a$-conversion, $`t=\mtyp{\k(\fiv)}{\D}$ for some~\D.
  Let $(\ep_3,\scond_3)=
  \algsubp{\C}{\D}{b}{\fc{\zero}{a}{\i(b,\fiv)}}{(a,\fiv)}{(a<\k,\ictx)}{-}$.
  Then $\ep_3$ is positively specified for~\D\ and negatively for~\A, and if~$\ep_3``(=\eptwo\vDash\scond_3$ then 
  \begin{displaymath}
    \sjudg[\eptwo]{a,\fiv}{a<\k,\ictx}{\C\isubst[b]{\fc{\zero}{a}{\i(b,\fiv)}}\teq\D}.
  \end{displaymath}
  Hence $\ep_1`U\ep_2`U\ep_3$ is positively defined for~\D\ and negatively for~\B, and for any~$\eptwo``)=\ep_1`U\ep_2`U\ep_3$, $\eptwo\vDash(\scond_1,\scond_2,\scond_3)$ implies
  \begin{math}
    \sjudg[\eptwo]{a,\fiv}{a<\k,\fc{\zero}{a}{\i(b,\fiv)}<\h(\fiv),\ictx}{
      \B\isubst{\zero}\isubst[b]{\fc{\zero}{a}{\i(b,\fiv)}}\teq\D}.
  \end{math}
  Hence for any ~$\eptwo``)=\ep_1`U\ep_2`U\ep_3`U\{\h(\fiv)=\fc{\zero}{\k(\fiv)}{\i(b,\fiv)}\}$ such that $\eptwo\vDash(\scond_1,\scond_2,\scond_3)$ and \ijudg[\eptwo]{\fiv}{\ictx}{\fc{\zero}{\k(\fiv)}{\i(b,\fiv)}`|},
  \begin{displaymath}
    \sjudg[\eptwo]{\fiv}{\ictx}{
      \mtyp{\k}{\B\isubst{\zero}\isubst[b]{\fc{\zero}{a}{\i(b,\fiv)}}}\teq\mtyp{\k}{\D}}.
  \end{displaymath}
  To deal with the context, we proceed as for~($\multimap$):
  for every $y:\mtyp[b_y]{\J_y}{\B_y}`:`G$, let $(\A_y,\ep_y,\scond_y)=\algdigp{\B_y}{\h(\fiv)}{\J_y}{\fiv}{b}{b_y}{\ictx}{-}$.
  By Lemma~\ref{lem:digging}, $\ep_y$ is positively specified for~$\B_y$ and negatively for~$\A_y$, and $\ep_y``(=\eptwo\vDash\scond_y$ implies
  \begin{displaymath}
    \sjudg[\eptwo]{\fiv}{\ictx}{\mtyp[b_y]{\sum_{b<\h(\fiv)}\J_y}{\A_y}\teq
      \sum_{b<\h(\fiv)}\mtyp[b_y]{\J_y}{\B_y'}}
  \end{displaymath}
  for some $\B_y'$ such that
  $\sjudg[\eptwo]{\fiv}{\ictx,b<\h(\fiv),b_y<\J_y}{\B_y'\teq\B_y}$.
  Let~$\l_y(\fiv)$ be a fresh symbol, and write $\ep_y'=\ep_y\cup\{\l_y(\fiv)=\sum_{b<\h(\fiv)}\J_y\}$ and $\scond_y'=\scond_y\cup\{\ijudg[]{\fiv}{\ictx}{\sum_{b<\l(\fiv)}\J_y`|}\}$.

  Finally we return $(\der,\scond)$ where

  \hfil
  \begin{math}
    \begin{array}[t]{c@{~=~}l}
      \scond & \scond'`U\scond_1`U\scond_2`U\scond_3`U
      \{\ijudg[]{\fiv}{\ictx}{\fc{\zero}{\k(\fiv)}{\i(b,\fiv)}`|}\}`U
      \bigcup_y \scond_y' \\
      \der & \displaystyle
      \frac{\der'}{
        \judg{\fiv}{\ictx}{`D}{\h(\fiv)+\sum_{b<\h(\fiv)}\J}{
          \fix{\t}}{`t}}\\
      \ep & \ep'`U\ep_1`U\ep_2`U\ep_3`U
      \{\h(\fiv)=\fc{\zero}{\k(\fiv)}{\i(b,\fiv)}\}`U
      \bigcup_y\ep_y' \\
      `D & \{y:\mtyp[b_y]{\l_y(\fiv)}{\A_y}\}_y\\
    \end{array}
  \end{math}
}
\end{varitemize}
\begin{lemma}
  \label{lem:welldef-mainfct}
  For every \fiv, \ictx, and every \PCF\ derivation \derpcf, 
  $\mainfct(\fiv;\ictx;\derpcf)$ is well defined on the form $(\ljudg{`G}{\I}{t}{`t};\rep;\scond)$, and satisfies Decoration and Polarity.
\end{lemma}

\subsection{Correctness}
\label{sec:alg-cor}
The algorithm we have just finished describing needs to be proved sound and complete with respect to \dlpcfv\ typing. As usual, this is not
a trivial task. Moreover, linear dependent types have a semantic nature which makes the task of formulating (if not proving) the desired results
even more challenging.

\subsubsection{Soundness}
\label{sec:snd}
A type inference procedure is \emph{sound} when the inferred type can actually be derived by way of the type system at hand.
As already remarked, \mainfct\ outputs an equational program $\rep$ which possibly contains unspecified symbols and which, as a consequence, cannot be exploited in typing.
Moreover, the role of the set of proof obligations in \scond is maybe not clear at first.
Actually, soundness holds for \emph{every} completely specified $\ep\supseteq\rep$ which makes the proof obligations in $\scond$ true: 
\begin{theorem}[Soundness]
  \label{theo:algo-snd}
  If \derpcf\ is a \PCF\ derivation for~$\t$,
  then for any \fiv\ and \ictx,
  $\mainfct(\fiv;\ictx;\derpcf)=(\ljudg{`G}{\I}{\t}{`t};\rep;\scond)$
  where \judg[\rep]{\fiv}{\ictx}{`G}{\I}{\t}{`t} is correctly specified and for any $\ep\supseteq\rep$,
  \begin{displaymath}
    \scval{\scond}
    \quad\implies\quad
    \judg{\fiv}{\ictx}{`G}{\I}{\t}{`t}
    \text{ is derivable and precise}.
  \end{displaymath}
\end{theorem}
%

\subsubsection{Completeness}
\label{sec:cpl}
But are we sure that \emph{at least} one type derivation can be built from the outcome of \mainfct\ \emph{if} one such type derivation exists? Again, it
is nontrivial to formulate the fact that this is actually the case.
\begin{theorem}[Completeness]
  \label{theo:algo-cpt}
  If \judg[\ep]{\fiv}{\ictx}{`D}{\J}{\t}{`s} is a precise \dlpcfv\ judgement derivable with~\der, then $\mainfct(\fiv;\ictx;\forget{\der})$ is of the form 
  $(\ljudg{`G}{\I}{\t}{`t};\rep;\scond)$, and there is $\mapping:\ep\models\rep$ such that $\scval[\mapping]{\scond}$.
\end{theorem}
A direct consequence of soundness and completeness (and the remark on Definition~\ref{def:valid-model}) is the following:
\begin{corollary}
  \label{cor:eq-typ-sc}
  If a closed term~\t\ is typable in \PCF\ with type \NatPCF\ and a derivation~\derpcf, then 
  \begin{displaymath}
    \mainfct(∅;∅;\derpcf)~=\quad(\ljudg{}{\I}{\t}{\NatU{\f}};\ep;\scond)
  \end{displaymath}
  and~\t\ is typable in \dlpcfv\ iff \scval{\scond}.
\end{corollary}



\section{Type Inference at Work}
\label{sec:side_cond}

The type inference algorithm presented in the previous section has been implemented in \textsc{Ocaml}\footnote{the 
source code is available at \url{http://lideal.cs.unibo.it}.}.
Programs, types, equational programs and side conditions become values of appropriately defined inductive
data structures in \textsc{Ocaml}, and the functional nature of the latter makes the implementation effort
easier. This section is devoted to discussing the main issues we have faced along the process, which is still ongoing.

The core of our implementation is an \textsc{Ocaml} function called \progname.
Taking a closed term~\t\ having \PCF\ type \tpcfone\ in input, \progname\ returns 
a typing derivation~\der, an equational program~\ep\ and a set of side conditions~\scond. The conclusion of~\der\ is a 
\dlpcf\ typing judgement for the input term. If \tpcfone\ is a first-order type, then the produced judgement
is derivable \textit{iff} all the side conditions in \scond\ are valid~(see Corollary~\ref{cor:eq-typ-sc}).
To do so, \progname\ calls (an implementation of) \mainfct\ on \t\ and a context $\fiv$ consisting of $n$ unconstrained index variables, where $n$ is the arity of $\t$.
This way, \progname\ obtains \rep\ and \scond\ as results, and then proceeds as follows:
\begin{varitemize}
\item
  If $\tpcfone$ is $\NatPCF$, then \rep\ is already completely specified and Corollary~\ref{cor:eq-typ-sc} ensures that we already have
  what we need.
\item
  If  $\tpcfone$ has a strictly positive arity, then some of the symbols in \rep\ are unspecified, and appropriate equations for them need to be added to \rep.
  Take for instance a term~\s\ of type $\NatPCF\toPCF\NatPCF$.
  \progname(\s) returns~\rep, \scond, and a typing judgement on the form
  \begin{displaymath}
    \ejudg[\rep]{a;\emptyset}{\K}{\t}{\NatU{\g(a,b)}\marr{b<\f(a)}\NatU{\j(a,b)}},
  \end{displaymath}
  where~\j\ is a positive symbol while $\f$ and $\g$ are negative, thus unspecified in $\rep$.
  $\rep$ can be appropriately ``completed'' by adding the equations $\f(a):=\inb[1]$ and $\g(a,b):=a$ to it.
  This way, we are insisting on the behaviour of $\t$ when fed with \emph{any} natural number (represented by~$a$) and when the environment needs $\t$ only once.
\end{varitemize}
How about complexity analysis? Actually, we are already there: the problem of proving the number of
machine reduction steps needed by \t\ to be at most $p:\mathbb{N}\rightarrow\mathbb{N}$ (where $p$ is, e.g. a polynomial)
becomes the problem of checking $\scval{\scondtwo}$ where \ep\ is the appropriate completion of~\rep, and $\scondtwo$ is $\scond\cup\{(\K+1)(|t|+2)\leq p(a)\}$ (Proposition~\ref{prop:int-snd}).

\paragraph{Simplifying Equations.}
Equational programs obtained in output from \progname\ contains many equations which are trivial 
(such as $\f(a)=\inb$ or $\f(a)=\g(a)$), and as such can be eliminated. Moreover, instances of forest
cardinalities and bounded sums can sometime be greatly simplified. As an example, $\sum_{a< \inb[0]}\J$
can always be replaced by $\inb[0]$. This allows, in particular, to turn $\rep$ into a set of
fewer and simpler rules, thus facilitating the next phase.

A basic simplification procedure has already been implemented, and are called by \progname\ on the output 
of \mainfct.
However, automatically treat the equational program by an appropriate prover would be of course desirable.
For this purpose, the possibility for \progname\ to interact with \textsc{Maude}~\cite{maude}, a system supporting equational 
and rewriting logic specification, is currently investigated.

\paragraph{Checking Side Conditions.}
As already stressed, once \progname\ has produced a pair $(\rep,\scond)$, the task we started from, namely complexity analysis of $\t$, is not finished, yet:
checking proof obligations in $\scond$ is as undecidable as analysing the complexity of $\t$ directly, since most of the obligations in $\scond$ are termination statements anyway.
There is an important difference, however:
statements in $\scond$ are written in a language (the first-order equational logic) which is more amenable to be treated by already existing automatic and semi-automatic tools.

Actually, the best method would be to first call as many existing automatic provers on the set of side conditions, then asking the programmer to check those 
which cannot be proved automatically by way of an interactive theorem prover.
For this purpose, we have implemented an algorithm translating a pair in the form $(\rep,\scond)$ into a \textsc{Why3}~\cite{why3} theory.
\lv{
  Indeed, \textsc{Why3} is an intermediate tool between first order logic and various theorem provers, from SMT solvers to the \textsc{Coq} proof assistant~\cite{coq}\footnote{
  \textsc{Alt-ergo}, \textsc{Cvc3}, \textsc{E-prover}, \textsc{Gappa},   \textsc{Simplify}, \textsc{Spass}, \textsc{Vampire}, \textsc{veriT}, \textsc{Yicex} and \textsc{Z3}}
  on the side conditions produced by $\progname(\t)$. Most of them are actually proved automatically, at least in the few 
  example programs we have mentioned in the course of this paper~\footnote{The results of these tests are available at \url{http://lideal.cs.unibo.it/}}.
  Actually, the symbol names used in the \textsc{Why3} theory are the same as the ones used for the annotation of the type derivation (that can be printed by \progname).
  Hence trying to check interactively some side conditions, the programmer can access both the definition of a symbol in the equational program, and the subterm of~\t\ 
  it refers to in the type derivation.
  
  However, using the proof assistant \textsc{Coq} on the few side conditions that remain to be interactively checked \textit{through} the \textsc{Why3} tool is not as simple as it should be.
  This is due to the way we express bounded sum and forest cardinality indexes in a first order logic.
  To facilitate the work of the programmer checking the side conditions, it would thus be suitable to also translate them directly into \textsc{Coq}, 
  making use of its higher order definition facilities.  The interactive theorem prover would be called \emph{directly}, but only on those side conditions that cannot 
  be proved automatically. This of course requires some special care: we would like to preserve a formal link between the \textsc{Why3} theory and 
  the statement which that have to be proved in \textsc{Coq} (or in any interactive theorem prover).
  
  
  All these issues are currently investigated and developed within \textsc{Lideal}~\cite{lideal}.}



\section{Related Work}
\label{sec:rw}
Complexity analysis of higher-order programs has been the object of study of much research.
We can for example mention the many proposals for type systems for the $\lambda$-calculus
which have been shown to correspond in an \emph{extensional sense} to, \eg\ polynomial time
computable functions as in implicit computational complexity. Many of them can be seen
as static analysis methodologies: once a program is assigned a type, an upper bound to
its time complexity is relatively easy to be synthesised. The problem with these systems,
however, is that they are usually very weak from an \emph{intentional} point of view, since
the class of typable programs is quite restricted compared to the class of all terms working
within the prescribed resource bounds.

More powerful static analysis methodologies can actually be devised.
All of them, however, are limited either to very specific forms of resource bounds or to a peculiar form
of higher-order functions or else they do not get rid of higher-order as the underlying logic.  
Consider, as an example, one of the earliest work in this direction, namely Sands's system
of cost closures~\cite{Sands90}: the class of programs that can be handled includes
the full lazy $\lambda$-calculus, but the way complexity is reasoned about remains genuinely
higher-order, being based on closures and contexts. In Benzinger's framework~\cite{Benzinger04}
higher-order programs are translated into higher-order equations, and the latter are turned into
first-order ones; both steps, and in particular the second one, are not completeness-preserving. 
Recent works on amortised resource analysis are either limited to first-order programs~\cite{HAH11} or to linear
bounds~\cite{HOAA10}. A recent proposal by Amadio and R\'egis-Gianas~\cite{AmadioRG} allows to reason
on the the cost of higher-order functional programs by way of so-called cost-annotations,
being sure that the actual behaviour of compiled code somehow reflects the annotation. The
logic in which cost annotations are written, however, is a form of \emph{higher-order} Hoare logic. 
None of the proposed systems, on the other hand, are known to be (relatively) complete in the
sense we use here.

Ghica's slot games~\cite{Ghica05} are maybe the work which is closest to ours, among the many in the literature.
Slot Games are simply ordinary games in the sense of game semantics, which are however instrumented
so as to reflect not only the observable behaviour of (higher-order) programs, but also their
performance. Indeed, slot games are fully abstract with respect to an operational theory 
of improvements due do Sands~\cite{Sands91}: this can be seen as the counterpart of our relative completeness
theorem. An aspect which has not been investigated much since Ghica's proposal is whether slot games
provides a way to perform actual verification of programs, maybe via some form of model checking. 
As we have already mentioned, linear dependency can be seen as a way to turn games and strategies
into types, so one can see the present work also as an attempt to keep programs and strategies
closer to each other, this way facilitating verification.


\section{Conclusions}
\label{sec:concl}

A type inference procedure for \dlpcf\ has been introduced which, given a \PCF\ term, reduces the
problem of finding a type derivation for it to the one of solving proof obligations on
an equational program, itself part of the output. Truth of the proof obligations correspond to termination
of the underlying program. Any type derivation in \dlpcf\ comes equipped with an expression bounding
the complexity of evaluating the underlying program. Noticeably, proof obligations and the related
equational program can be obtained in polynomial time in the size of the input \PCF\ program.

The main contribution of this paper consists in having shown that linear dependency is not only a very
powerful tool for the precise analysis of higher-order functional programs, but is also a way to 
effectively and efficiently turn a complex problem (that of evaluating the time complexity of an higher-order program)
into a much easier one (that of checking a set of proof obligations for truth).

Although experimental evaluation shows that proof obligations can potentially be handled by modern
tools mixing automatic and semi-automatic reasoning, as explained in Section~\ref{sec:side_cond}, much remains
to be done about the technical aspects of turning proof obligations into a form which is suitable to
automatic or semi-automatic solving. Actually, many different tools could conceivably be of help here, each
of them requiring a specific input format.
This implies, however, that the work described here, although not providing a fully-fledged out-of-the-box methodology, has the merit of allowing to factor a complex non-well-understood
problem into a much-better-studied problem, namely verification of first-order inequalities on the natural
numbers.


\lvsv{\bibliographystyle{abbrv}}{\bibliographystyle{splncs03}}
\bibliography{bibli}

\end{document}